\documentclass[11pt,reqno,a4paper,nofootinbib]{article}

\usepackage{amsmath}
\usepackage{amsthm}
\usepackage{amssymb}
\usepackage{dsfont}
\usepackage{graphicx}
\usepackage[nottoc,numbib]{tocbibind}
\usepackage{authblk}

\newtheorem{theorem}{Theorem}
\newtheorem{proposition}[theorem]{Proposition}
\newtheorem{lemma}[theorem]{Lemma}
\newtheorem{corollary}[theorem]{Corollary}
\newtheorem{conjecture}[theorem]{Conjecture}

\theoremstyle{definition}
\newtheorem{remark}[theorem]{Remark}

\newcommand{\binH}{H}
\newcommand{\binrel}{D_2}
\newcommand{\var}{{\rm var}}
\newcommand{\ket}[1]{|#1\rangle}
\newcommand{\bra}[1]{\langle#1|}
\newcommand{\tr}[1]{{\rm{tr}}\left[#1\right]}
\newcommand{\1}{\mathds{1}}
\newcommand{\R}{\mathbb{R}}
\newcommand{\supp}{{\rm supp}}
\newcommand{\diag}{{\rm diag}}
\newcommand{\ii}{\mathds{1}}
\newcommand{\C}{\mathbb{C}}

\title{\vspace{-1.2cm}{\Large\textbf{Tight bound on relative entropy by entropy difference}}}
\author{\vspace{-0.15cm}David Reeb\thanks{reeb.qit@gmail.com}~}
\author{Michael M.~Wolf\thanks{m.wolf@tum.de}}
\affil{\vspace{-0.15cm}\small{Department of Mathematics, Technische Universit\"at M\"unchen, 85748 Garching, Germany}}
\date{}

\begin{document}

\maketitle

\vspace{-1cm}\begin{abstract}We prove a lower bound on the relative entropy between two finite-dimensional states in terms of their entropy difference and the dimension of the underlying space. The inequality is tight in the sense that equality can be attained for any prescribed value of the entropy difference, both for quantum and classical systems. We outline implications for information theory and thermodynamics, such as a necessary condition for a process to be close to thermodynamic reversibility, or an easily computable lower bound on the classical channel capacity. Furthermore, we derive a tight upper bound, uniform for all states of a given dimension, on the variance of the surprisal, whose thermodynamic meaning is that of heat capacity.
\end{abstract}

\medskip
\medskip

\tableofcontents

\newpage

\section{Introduction}\label{introsection}
The relative entropy is a distance-like measure that appears in a multitude of areas, such as information theory, thermodynamics, statistics and learning theory, being of operational significance in various situations (see Section \ref{applicationsection} for a few applications). Also known as the Kullback-Leibler divergence, it was first introduced for probability distributions \cite{kullbackleibler}, and later generalized to quantum states \cite{umegaki}. Another ubiquitous quantity is the entropy of a probability distribution or quantum state \cite{shannonmathemtheorycommun,vonNeumannFoundationsQM}, which in Thermodynamics had already played a central role because entropy differences characterize possible and impossible thermodynamic state transformations (see e.g.\ the Clausius inequality in Section \ref{stepwiseequlibrationsubsect}).

In this work, we provide a lower bound on the relative entropy $D(\sigma\|\rho)$ between two states $\sigma$, $\rho$ (probability distributions or quantum states) in terms of their entropy difference $\Delta=S(\sigma)-S(\rho)$. Qualitatively, it is clear that such non-trivial lower bounds exist in any finite dimension due to the compactness of the state space, since $\Delta\neq0$ implies $\sigma\neq\rho$ and thus $D(\sigma\|\rho)>0$ by Klein's inequality \cite{ohyapetz}. Our main inequality (Theorem \ref{maintheoremineqpaper}) makes this quantitative and is furthermore \emph{tight}, meaning that for each dimension $d$ it provides the best lower bound on $D(\sigma\|\rho)$ in terms of $\Delta$, both for classical and quantum systems. We note that any lower bound that can be derived by combining the tight Pinsker inequality \cite{pinskersinequality,refinementspinskier,audenaerteisertoptimal} with the tight Fannes-Audenaert inequality \cite{fannesinequalitypaper,audenaertfannes,fannesaudenaertelegant} will \emph{not} be tight and will be strictly weaker than the derived bounds, even in its functional dependence (see Remark \ref{pinskerfannesaudenaertremark}).

Also considering states of finite dimension $d$, in Section \ref{mainresults2} we give a tight upper bound on the variance of the \emph{surprisal} (or information gain), which is quadratic in $\log d$ (Section \ref{subsectionvarsurprisal}); of course, the expectation value of the surprisal is just the entropy and is bounded by $\log d$ \cite{ohyapetz}. One thermodynamic implication of this result is an upper bound on the heat capacity of finite-dimensional systems (Section \ref{maxheatcapacitysubsection}).

The main results of the present paper thus contribute new items to the set of dimension-dependent entropy bounds, of which the Fannes-Audenaert inequality is the single most well-known, and arguably most important, instance within information theory.

\bigskip

The inequalities presented here arose out of, and are used in, an investigation of finite-size effects in Landauer's Principle \cite{longlandauercompanionpaper}, but we expect them to have applications elsewhere in thermodynamics and information theory; some are outlined in Section \ref{applicationsection}. Furthermore, the finite-size bounds here arise in one-partite systems, whereas the Landauer scenario -- the topic of \cite{longlandauercompanionpaper} -- is bipartite, involving a system and a thermal reservoir \cite{landauer}.

Physically, our bounds are especially interesting for quantum thermodynamics \cite{gemmerbook,popescureachlimit} and generally for the thermodynamics of microscopic systems or devices. Furthermore, even a large heat bath may sometimes be reasonably treated as small, when the equilibration time with another system is so short that only a small part of the bath effectively interacts with the system.  Our bounds can be applied to derive finite-size corrections to well-known physical laws and, for example, alter efficiency analyses of physical process like Carnot's or Landauer's \cite{andersgiovannetti,longlandauercompanionpaper}.

By treating the Shannon and von Neumann (relative) entropies, our results are relevant to the conventional situation of many independent copies of a system state (``thermodynamic limit''), averaging quantities over these copies (``ensemble averages''). Thermodynamics and information theory can instead also be examined in the ``single-shot setting'', necessitating extra parameters such as the success probability of a process (e.g.\ \cite{rennerthesis,trulyworklike,oppenhoro,egloff}). Our setup  is thus different from the one-shot scenario:\ whereas the latter concerns a \emph{finite (small) number} of systems, our results have implications in the limit of infinitely many \emph{finite-dimensional} system copies. The variance computed in Section \ref{subsectionvarsurprisal}, however, can quantify how many copies of a finite-dimensional system have to be averaged before the Shannon or von Neumann entropies become sensible measures (see also \cite{tomamichaelvariance,kelisecondorder}).

\subsection{Notation}\label{notationsubsection}
All states $\rho$, $\sigma$ will be on a space of finite dimension $d<\infty$. In the quantum framework, states are positive semi-definite $d\times d$-matrices of trace $1$ (``density matrices'' \cite{nielsenchuang}). In the classical (probability theory) framework, they are probability distributions on $d$ atomic events \cite{coverthomas}. For a unified presentation of our results in both the classical and quantum setups, we will throughout identify such probability distributions with density matrices of size $d\times d$ that are diagonal w.r.t.\ a fixed basis and have the $d$ atomic probabilities $p_1,\ldots,p_d$ as diagonal entries; the notation $\rho={\rm diag}(p_1,\ldots,p_d)$ provides the translation between both domains. We often require $d\geq2$ to exclude the trivial one-dimensional case, in which some statements become pathological.

The \emph{entropy} of a state $\rho$ is defined as
\begin{align}\label{entropydefinition}
S(\rho)~:=~-\tr{\rho\log\rho}~.
\end{align}
Throughout, we use the natural logarithm, denoted by $\log$, and employ the usual rules of calculus on the extended real line $\overline{\R}:=\R\cup\{\pm\infty\}$, such as $0\log0:=0$; only in Section \ref{wrongcodesubsubsection} will we also use the $D$-ary logarithm $\log_Dx:=(\log x)/(\log D)$, with $D>1$. A quantity of central interest will be the \emph{entropy difference} $\Delta\equiv\Delta(\sigma,\rho)$ of the states $\sigma$ and $\rho$:
\begin{align}\label{definenotationDelta}
\Delta(\sigma,\rho)~:=~S(\sigma)-S(\rho)\,\in[-\log d,+\log d]~.
\end{align}

The other central quantity is the \emph{relative entropy} between two states $\sigma$ and $\rho$:
\begin{align}
D(\sigma\|\rho)~:=~\tr{\sigma\log\sigma}-\tr{\sigma\log\rho}~,
\end{align}
which equals $+\infty$ if $\supp[\sigma]\not\subseteq\supp[\rho]$, and is finite otherwise, non-negative, and vanishes iff $\sigma=\rho$.

We also define binary versions of the entropy and relative entropy, i.e.\ for binary probability distributions $(x,1-x)$ and $(y,1-y)$ with $0\leq x,y\leq1$:
\begin{align}
\binH(x)~&:=~S\left(\diag(x,1-x)\right)~&=&~x\log\frac{1}{x}+(1-x)\log\frac{1}{1-x}~,&\\
\binrel(x\|y)~&:=~D\left(\diag(x,1-x)\|\diag(y,1-y)\right)~&=&~x\log\frac{x}{y}+(1-x)\log\frac{1-x}{1-y}~.&\label{definebinaryrelent}
\end{align}

Note that the entropy difference $\Delta(\sigma,\rho)$ changes sign under exchange of $\sigma$ and $\rho$, whereas the relative entropy $D(\sigma\|\rho)$ does not generally have any symmetry under exchange. For example, $\Delta=-\log d$ forces $\rho$ to be the maximally mixed state $\1/d$ and $\sigma$ to be any pure state (any Hermitian projector of rank $1$), resulting in $D(\sigma\|\rho)=\log d$; whereas $\Delta=+\log d$ interchanges these $\rho$ and $\sigma$ and gives $D(\sigma\|\rho)=\infty$. The latter case is special as for any other $\Delta\in[-\log d,\log d)$ there exist full-rank states $\sigma$ and $\rho$ with $\Delta(\sigma,\rho)=\Delta$, such that $D(\sigma\|\rho)<\infty$ is finite.

For a more detailed discussion of entropic quantities we refer to \cite{ohyapetz} and \cite{wehrlreview} or, in the context of classical and quantum information theory, to \cite{coverthomas} and \cite{nielsenchuang}.

The acronyms LHS and RHS mean ``left-hand side'' and ``right-hand side'', respectively.

\section{Main results}\label{mainresultssection}

In Section \ref{relentvsentdifferencesubsection} we state the tight inequality between relative entropy and entropy difference (Theorem \ref{maintheoremineqpaper}) and describe properties and simplifications of the bound (Theorem \ref{propertiestheorem} and Remarks \ref{equalitycasesremark}--\ref{trivialdimindepbounds}) which are useful for applications (see Section \ref{applicationsection}). The tight upper bound on the variance of the surprisal (or heat capacity) is given in Section \ref{mainresults2}. The proofs follow in Section \ref{proofsection}.

\subsection{Relative entropy vs.\ entropy difference}\label{relentvsentdifferencesubsection}
To state our main inequality and its simplifications, we define for $d\geq2$ and $\Delta\in[-\log d,\log d]$:
\begin{align}
M(\Delta,d)~&:=~\min_{0\leq s,r\leq(d-1)/d}\,\left\{\binrel(s\|r)\,\big|\,\binH(s)-\binH(r)+(s-r)\log(d-1)=\Delta\,\right\}~,\label{definefunctionM}\\
N(d)~&:=~\max_{0<r<1/2}r(1-r)\,\left(\log\left(\frac{1-r}{r}(d-1)\right)\right)^2~,\label{definerealN}\\
N_d~&:=~\frac{1}{4}\log^2(d-1)+1~.\label{defineapproxN}
\end{align}
(The expression $\log^2(d-1)$ should always be read as $\left(\log(d-1)\right)^2$.) All of these quantities can be efficiently computed numerically as they involve optimizations over at most two bounded real variables. See also Fig.\ \ref{Mfigure}, and Lemmas \ref{lemmathatreducesoptimizationovertwovariables} and \ref{lemmareducingonevariableoptimization} (Section \ref{technicallemmassubsection}) for relations among (\ref{definefunctionM})--(\ref{defineapproxN}).

\bigskip

\begin{theorem}[Tight lower bound on relative entropy by entropy difference]\label{maintheoremineqpaper}
Let $\sigma$, $\rho$ be states of dimension $d$, with $2\leq d<\infty$, and define $\Delta:=S(\sigma)-S(\rho)$. Then:
\begin{align}
D(\sigma\|\rho)~\geq~M(\Delta,d)~,\label{ineqalityinmaintheorem}
\end{align}
with the function $M(\Delta,d)$ defined in Eq.\ (\ref{definefunctionM}).

Conversely, for any $\Delta\in[-\log d,\log d]$, there exist $\sigma$, $\rho$ attaining equality in (\ref{ineqalityinmaintheorem}). More precisely, for any pair $(s,r)$ attaining the minimum in (\ref{definefunctionM}), the commuting $d$-dimensional states
\begin{align}\label{formofsigmaandrhoinDvsDeltaLemma}
{\sigma}~:=~{\rm diag}\left(1-{s},\frac{{s}}{d-1},\ldots,\frac{{s}}{d-1}\right)~,\quad{\rho}~:=~{\rm diag}\left(1-{r},\frac{{r}}{d-1},\ldots,\frac{{r}}{d-1}\right)
\end{align}
have entropy difference $S(\sigma)-S(\rho)=\Delta$ and achieve equality $D(\sigma\|\rho)=M(\Delta,d)$.
\end{theorem}

\smallskip

\begin{theorem}[Properties of the tight bound]\label{propertiestheorem}
Let $2\leq d<\infty$. Then the function $M(\Delta,d)$ in the tight lower bound (\ref{ineqalityinmaintheorem}) is non-negative, continuous and strictly convex in $\Delta\in[-\log d,\log d]$, and continuously differentiable in the interior of this interval. It takes values $M(0,d)=0$, $M(-\log d,d)=\log d$, $M(\log d,d)=\infty$, and $M(\Delta,d)<\infty$ for $\Delta\in[-\log d,\log d)$.

For any $N\geq N(d)$, with $N(d)$ from Eq.\ (\ref{definerealN}), the following lower bounds hold for all $\Delta$:
\begin{align}
M(\Delta,d)~&\geq~N\left(e^\frac{\Delta}{N}-1-\frac{\Delta}{N}\right)~\geq~\frac{\Delta^2}{2N}+\frac{\Delta^3}{6N^2}~,\label{lowerboundeqninmaintheorem}\\
M(\Delta,d)~&\geq~\frac{\Delta^2}{3\log^2d}~.\label{quadraticlowerboundforallDelta}
\end{align}
Easily computable choices for $N$ are $N=N_d=\frac{1}{4}\log^2(d-1)+1>N(d)$, or $N=\log^2d>N(d)$.
\end{theorem}

\medskip

\begin{remark}[Equality cases in Eq.\ (\ref{ineqalityinmaintheorem})]\label{equalitycasesremark}
Regarding the equality statement in Theorem \ref{maintheoremineqpaper}, we remark that for any $\Delta\in[-\log d,\log d]$ the minimum in (\ref{definefunctionM}) actually exists, i.e.\ is attained for some pair $(s,r)$ (see Section \ref{proofofpropertiessubsect}), and equals $\infty\in\overline{\R}$ for $\Delta=\log d$. Note that for states of the form (\ref{formofsigmaandrhoinDvsDeltaLemma}), it is $D(\sigma\|\rho)=\binrel(s\|r)$, $S(\sigma)=\binH(s)+s\log(d-1)$ and similar for $S(\rho)$. In Remark \ref{exponentialfamilyremark} we elaborate on the states (\ref{formofsigmaandrhoinDvsDeltaLemma}), which come all from the same exponential family.

For $\Delta\neq0$, the pair $(\sigma,\rho)$ from (\ref{formofsigmaandrhoinDvsDeltaLemma}) constitutes, up to simultaneous unitary equivalence, the unique $d$-dimensional states achieving equality $D(\sigma\|\rho)=M(\Delta,d)$ and $S(\sigma)-S(\rho)=\Delta$. This follows from the proof of Theorem \ref{maintheoremineqpaper} in Section \ref{proofofmainthmsubsect} as the optimal states for $\Delta\neq0$ are necessarily of the form (\ref{formofsigmaandrhoinDvsDeltaLemma}) with $0\leq s,r\leq(d-1)/d$, and since for $\Delta\neq0$ the pair $(s,r)$ attaining the minimum in (\ref{definefunctionM}) is unique (which is shown in our proof of the convexity of $M(\Delta,d)$ in Section \ref{proofofpropertiessubsect}). For $\Delta=0$, exactly the pairs with $\sigma=\rho$ attain equality in (\ref{ineqalityinmaintheorem}).

As inequality (\ref{ineqalityinmaintheorem}) is tight for \emph{commuting} density matrices, it is tight for classical probability distributions (diagonal density matrices) as well.
\end{remark}

\begin{figure}[t]
\includegraphics[scale=0.85]{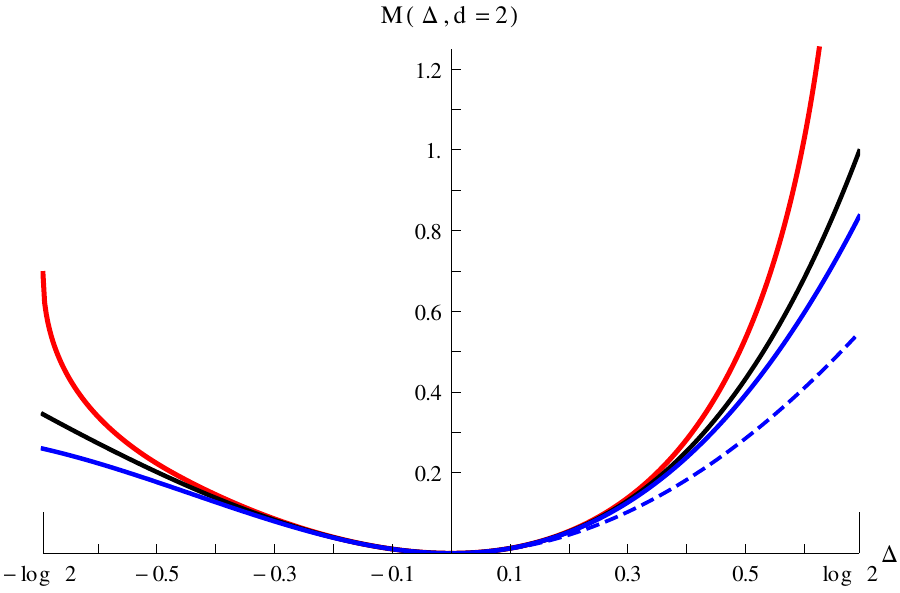}~~~~~~~\includegraphics[scale=0.85]{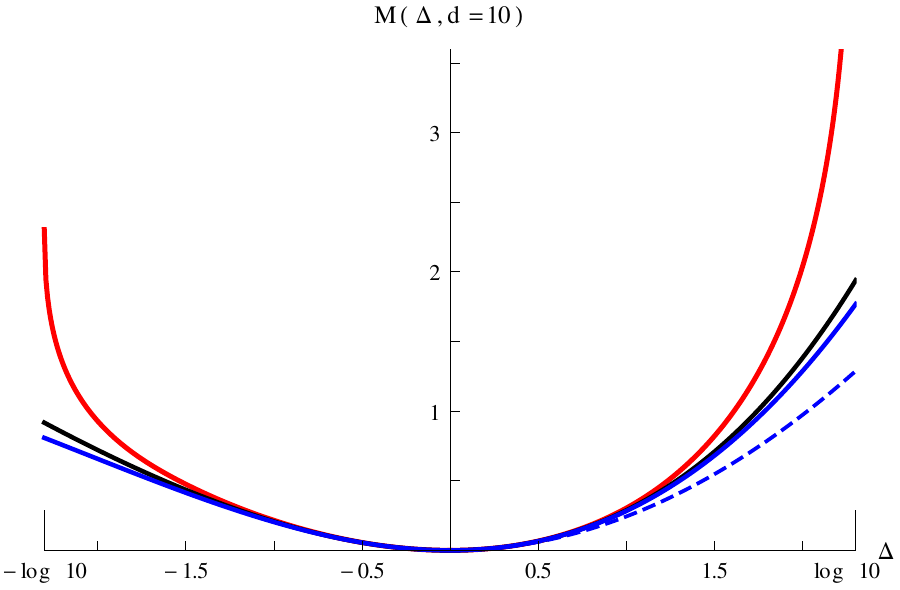}\\~\\
\includegraphics[scale=0.85]{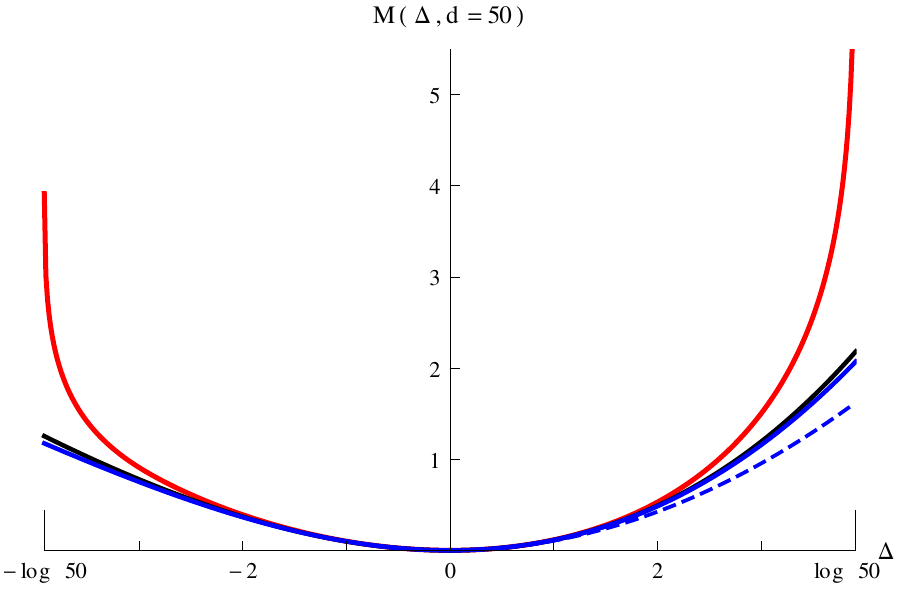}~~~~~~~\includegraphics[scale=0.85]{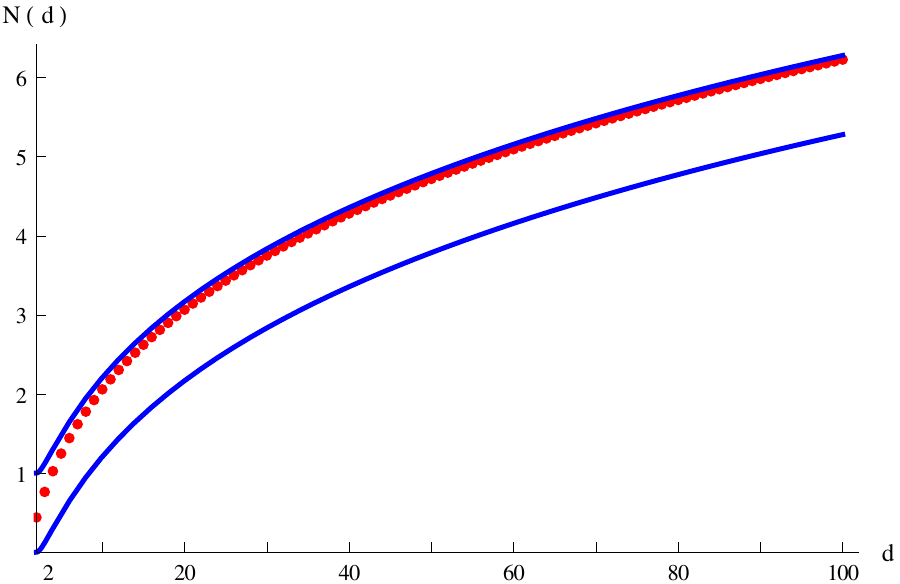}
\caption{\label{Mfigure}\emph{Upper and left panels:}\ The upper (red) curves show $M(\Delta,d)$ from Eq.\ (\ref{definefunctionM}) (tight lower bound of Theorem \ref{maintheoremineqpaper}) for $d=2,10,50$. The black and blue solid curves below are the lower bounds from Eq.\ (\ref{lowerboundeqninmaintheorem}) with the optimal $N=N(d)$, the dotted blue curve is the quadratic lower bound $\Delta^2/2N(d)$ (for $\Delta\geq0$). At $\Delta=\pm\log d$, all these lower bounds approach $2$ in the limit $d\to\infty$, whereas $M(-\log d,d)=\log d$ and $M(\log d,d)=\infty$. \emph{Lower right panel:}\ The red dots show $N(d)$ for $2\leq d\leq100$ (Eq.\ (\ref{definerealN})), which approaches its easily computable upper bound $N_d$ (Eq.\ (\ref{defineapproxN}) and Remark \ref{heatcapacityboundistight}) as $d\to\infty$ and which is bounded from below by $N_d-1$ (blue curves).}
\end{figure}

\begin{remark}[Goodness of the lower bounds in Eq.\ (\ref{lowerboundeqninmaintheorem})]\label{taylorexpansionremark}
One can check that the function $M(\Delta,d)$ is smooth around $\Delta=0$ (see Section \ref{proofofpropertiessubsect}) and that the RHS of (\ref{lowerboundeqninmaintheorem}) with $N=N(d)$ is its cubic Taylor expansion. This is thus the best cubic lower bound possible, and $\Delta^2/2N(d)$ is the best quadratic lower bound for $\Delta\geq0$ (Fig.\ \ref{Mfigure}); it is however not a lower bound for small $\Delta<0$.

The lower bounds in (\ref{lowerboundeqninmaintheorem}) are quite good (cf.\ Fig.\ \ref{Mfigure}) even for relatively large $|\Delta|$. For any constant $t\in(-1,1)$, the states (\ref{formofsigmaandrhoinDvsDeltaLemma}) with $s=(1+t)/2$, $r=(1-t)/2$ give $\Delta=\Delta(\sigma,\rho)=t\log(d-1)$ and $M(\Delta,d)\leq D(\sigma\|\rho)=\binrel(s\|r)=t\log(1+t)/(1-t)$, whereas the lowest bound in (\ref{lowerboundeqninmaintheorem}) gives $\simeq2t^2$ (the cubic term vanishes as $\sim1/\log d$). Even for the large values $\Delta=\pm\frac{1}{2}\log d$, the lower bound is thus tight (for large $d$) up to at most $10\%$.

The quantity $N(d)$ from Theorem \ref{propertiestheorem} appears in the upper bound in Theorem \ref{maxvarianceinfothm} as well, see Remark \ref{heatcapacityboundistight}. For this quantity, see also the lower right panel in Fig.\ \ref{Mfigure}.
\end{remark}

\begin{remark}[Monotonicity of $M(\Delta,d)$ in $\Delta$]\label{Mincreasingremark}
The tight lower bound $M(\Delta,d)$ is strictly monotonically decreasing in $\Delta$ in the regime $\Delta\leq0$, and strictly increasing in $\Delta$ in the regime $\Delta\geq0$ (cf.\ Fig.\ \ref{Mfigure}). This follows since the non-negative function $M(\Delta,d)$ vanishes at $\Delta=0$ and is strictly convex by Theorem \ref{propertiestheorem} (see also Fig.\ \ref{Mfigure}). As our convexity proof in Section \ref{proofofpropertiessubsect} is quite involved, we give now a simpler proof of monotonicity. We actually prove
\begin{align}\label{verystrictmonotonicityeqn}
M(\lambda\Delta,d)~<~\lambda\,M(\Delta,d)\qquad\text{for}~\Delta\in[-\log d,\log d]\setminus\{0\}\,,~\lambda\in(0,1)~.
\end{align}

Let first $\Delta\in(0,\log d)$, $\lambda\in(0,1)$, and let $\sigma$, $\rho$ be states with $D(\sigma\|\rho)=M(\Delta,d)$ and $S(\sigma)-S(\rho)=\Delta$. Define states $\sigma_\mu:=\mu\sigma+(1-\mu)\rho$ for $\mu\in[0,1]$. As $S(\sigma_\mu)$ is continuous in $\mu$, there exists $\mu'\in(0,1)$ with $S(\sigma_{\mu'})-S(\rho)=\lambda\Delta$, and by strict concavity of the entropy we have
\begin{align}
\lambda\Delta~>~\mu'S(\sigma)+(1-\mu')S(\rho)-S(\rho)~=~\mu'\Delta~,
\end{align}
i.e.~$\mu'<\lambda$. Convexity of the relative entropy \cite{ohyapetz} finally gives
\begin{align}
M(\lambda\Delta,d)~\leq~D(\sigma_{\mu'}\|\rho)~\leq~\mu'D(\sigma\|\rho)+(1-\mu')D(\rho\|\rho)~<~\lambda M(\Delta,d)~.
\end{align}
(\ref{verystrictmonotonicityeqn}) holds for $\Delta=\log d$ as well, since $M(\lambda\Delta,d)<\infty$ due to $\lambda\Delta<\log d$. The proof for $\Delta<0$ is similar, now replacing $\rho$ by some state $\rho_{\mu'}=\mu'\rho+(1-\mu')\sigma$ (see also the end of the proof of Theorem \ref{maintheoremineqpaper} in Section \ref{proofofmainthmsubsect} for the case $\Delta=-\log d$).
\end{remark}

\begin{remark}[Lower bounds from the Fannes-Audenaert and Pinsker inequalities]\label{pinskerfannesaudenaertremark}
A weaker lower bound on the relative entropy $D(\sigma\|\rho)$ in terms of the entropy difference $\Delta=\Delta(\sigma,\rho)$, as in Theorem \ref{maintheoremineqpaper}, can be obtained by combining the Fannes-Audenaert \cite{fannesinequalitypaper,audenaertfannes} and Pinsker \cite{pinskersinequality} inequalities:\ writing $T:=\|\sigma-\rho\|_1/2$ for the trace distance (or total variation or statistical distance \cite{coverthomas}) between the states $\sigma$ and $\rho$, we have the bound \cite{fannesinequalitypaper,audenaertfannes,fannesaudenaertelegant}
\begin{align}\label{fannesaudenaertineq}
|\Delta|~=~|S(\sigma)-S(\rho)|~\leq~T\log(d-1)+\binH(T)~=:~h_d(T)~\leq~T\left(1+\log(d-1)+\log1/T\right)~,
\end{align}
the first inequality being tight, and the sharpened Pinsker bound \cite{pinskersinequality,coverthomas,HOT81,refinementspinskier,audenaerteisertoptimal}
\begin{align}\label{pinskerinequality}
D(\sigma\|\rho)~\geq~s(T)~\geq~2T^2~,
\end{align}
where $s:[0,1]\to[0,\infty]$ is a function \cite{refinementspinskier} such that the first inequality is tight (for any dimension $d\geq2$) and which is bounded from below by its quadratic Taylor expansion, $s(x)\geq2x^2$.

If now $\Delta\in[-\log d,\log d]$ is given, we can invert the function $\left.h_d\right|_{[0,(d-1)/d]}:[0,(d-1)/d]\to[0,\log d]$ from (\ref{fannesaudenaertineq}), or bound the inversion of its RHS from below, to get a lower bound on $T$:
\begin{align}\label{invertfannesaudenaert}
T~\geq~h_d^{-1}(|\Delta|)~\geq~\frac{e-1}{e}\,\frac{|\Delta|}{1+\log(d-1)-\log|\Delta|}~,
\end{align}
where the prefactor is $(e-1)/e\approx0.63$. Plugging either of this into (\ref{pinskerinequality}) yields a lower bound on $D(\sigma\|\rho)$. This approach, however, can never yield a quadratic lower bound $\sim\Delta^2$ near $\Delta=0$, as (\ref{ineqalityinmaintheorem}) and (\ref{lowerboundeqninmaintheorem})--(\ref{quadraticlowerboundforallDelta}) together do, since the tight lower bound $s(T)$ in (\ref{pinskerinequality}) is quadratic near $T=0$ and since $h_d$ from (\ref{fannesaudenaertineq}) does \emph{not} satisfy $h_d^{-1}(|\Delta|)\geq c(d)|\Delta|$ for any positive $d$-dependent constant $c(d)$. Numerically, one actually sees that, for all $d\geq2$ and $\Delta\neq0$, the lower bound obtained by plugging the RHS of (\ref{invertfannesaudenaert}) into the RHS of (\ref{pinskerinequality}) is worse than the RHS of (\ref{lowerboundeqninmaintheorem}) with $N=N(d)$ (and even worse than the quadratic lower bound $\Delta^2/2N(d)$ for $\Delta>0$).

Furthermore, this approach can only ever yield lower bounds that are invariant under $\Delta\mapsto-\Delta$ since the Fannes-Audenaert and Pinsker inequalities are both symmetric in $\sigma$ and $\rho$. The tight lower bound $M(\Delta,d)$ however does actually not have this invariance (see Fig.\ \ref{Mfigure}).
\end{remark}

\begin{remark}[Dimension-independent bounds are trivial]\label{trivialdimindepbounds}The non-trivial lower bounds (i.e., which are strictly positive for $\Delta\neq0$) on the relative entropy from Theorems \ref{maintheoremineqpaper} and \ref{propertiestheorem} depend explicitly on the dimension $d<\infty$. This has to be so as any dimension-independent bound will necessarily be trivial:\ setting $t:=\Delta/\log(d-1)$ in the states of Remark \ref{taylorexpansionremark}, with any constant $\Delta\in(-\infty,+\infty)$ and for large enough dimension $d$, gives $\Delta(\sigma,\rho)=\Delta$ and $D(\sigma\|\rho)=O\left(2\Delta^2/\log^2(d-1)\right)\to0$ as $d\to\infty$, so that $0$ is the best possible dimension-independent lower bound for any fixed value of $\Delta$; this also holds for states over infinite-dimensional Hilbert spaces. In this case, however, the lower bound $0$ is never attained for $\Delta\neq0$ as $D(\sigma\|\rho)=0$ would imply $\sigma=\rho$ \cite{ohyapetz,brattelirobinson}, and thus $\Delta=0$ (if the entropies $S(\sigma)$, $S(\rho)$ are at all defined).

We further remark that the optimal lower bound $M(\Delta,d)$ is a decreasing function of $d$, implying that the finite-size corrections in applications (see Section \ref{applicationsection}) will be smaller for larger systems. To see this, let $d'>d\geq2$, $\Delta\in[-\log d,\log d]$, and let $s,r$ be optimal variables when computing $M(\Delta,d)$ in (\ref{definefunctionM}). Now define $s':=s$, and find $r'$ such that the entropy difference $\Delta(\sigma',\rho')$ between $d'$-dimensional states $\sigma'$, $\rho'$ as in (\ref{formofsigmaandrhoinDvsDeltaLemma}) equals the given $\Delta$; if $\Delta\neq0$, $r'$ will be closer to $s'=s$ than $r$ is to $s$, such that $M(\Delta,d')\leq \binrel(s'\|r')\leq \binrel(s\|r)=M(\Delta,d)$ with strict inequality for $\Delta\neq0$.
\end{remark}

\bigskip

The main part of the proofs of Theorems \ref{maintheoremineqpaper} and \ref{propertiestheorem} consists in reducing the minimization of $D(\sigma\|\rho)$ over (quantum) states $\sigma$, $\rho$ with a fixed value of $\Delta(\sigma,\rho)=\Delta$ to the simpler minimization over two bounded real variables in (\ref{definefunctionM}). The first step in this reduction is a simple argument that the bound (\ref{ineqalityinmaintheorem}) for quantum states follows from the corresponding bound for classical probability distributions, i.e.\ for all states $\sigma$, $\rho$ that are both diagonal w.r.t.\ a fixed basis. We give the full proofs in Sections \ref{proofofmainthmsubsect}--\ref{technicallemmassubsection}.

\subsection{Dimension bounds on second moments}\label{mainresults2}
In Section \ref{subsectionvarsurprisal} we derive a tight upper bound on the variance of the surprisal in terms of the dimension of the underlying space. Translating to thermodynamics in Section \ref{maxheatcapacitysubsection}, this yields an upper bound on the second moment of the energy of thermal states or, equivalently, on the heat capacity of finite-dimensional systems.

The derived bounds have apparent connections to the relative entropy inequalities from Theorems \ref{maintheoremineqpaper} and \ref{propertiestheorem}. Namely, the optimal states are of the same form and the (optimal) bounds involve the same quantities (see Remarks \ref{heatcapacityboundistight} and \ref{exponentialfamilyremark}). Also, all the bounds are dimension-dependent and become trivial for infinite-dimensional spaces (cf.\ Remark \ref{trivialdimindepbounds}). Furthermore, the heat capacity bound of Corollary \ref{maxheatcapacity} is in fact used in \cite{longlandauercompanionpaper} in a bipartite situation to bound a relative entropy term from below in an indirect way, as the direct bound by Theorem \ref{maintheoremineqpaper} would necessarily depend on an undesired entropic quantity (i.e.\ one from the ``wrong'' subsystem).

\subsubsection{Maximum variance of the surprisal}\label{subsectionvarsurprisal}
In a classical random experiment described by a probability distribution $\rho=\diag(p_1,p_2,\ldots,p_d)$, the \emph{information gain} upon outcome $i$ is $(-\log p_i)$, which is the unique sensible information measure in the limit of many independent experiments \cite{shannonmathemtheorycommun,coverthomas}. Equivalently, the surprise about obtaining $i$ may be quantified by the \emph{surprisal} $(-\log p_i)$. The (Shannon) entropy (\ref{entropydefinition}) is the expectation value of the surprisal, $S(\rho)=\sum_ip_i(-\log p_i)=\langle-\log\rho\rangle_\rho$. In this section, we look at its second moment, i.e.\ the \emph{variance} or fluctuation of the surprisal:
\begin{align}\label{expressvariationofinformationwithpi}
\var_\rho(-\log\rho)~:=~\sum_ip_i(-\log p_i)^2-\left(\sum_ip_i(-\log p_i)\right)^2~=~\tr{\rho\left(-\log\rho-S(\rho)\right)^2}.
\end{align}

In classical coding theory, when the source signals are i.i.d.\ distributed according to the spectrum of $\rho$, optimal prefix codes assign a codeword length of roughly $\simeq(-\log p_i)$ to symbol $i$ \cite{coverthomas}. The expected codeword length is thus $\simeq S(\rho)$ with fluctuation $\simeq\sqrt{\var_\rho(-\log\rho)}$, which implies a certain fluctuation in the lengths of encoded messages. (This holds up to an overall factor logarithmic in the size of the code alphabet, see Section \ref{wrongcodesubsubsection}.) Similar second-order effects in hypothesis testing using only finitely many copies have recently been investigated in \cite{tomamichaelvariance,kelisecondorder}.

The above definitions in terms of a general density matrix $\rho$ are sensible in the quantum framework as well, and have similar interpretations \cite{schumachercoding,nielsenchuang,indeterminatelengthQcoding}. Note that $S(\rho)$ and $\var_\rho(-\log\rho)$ depend both only on the eigenvalues of the density matrix $\rho$.

Our main theorem here places a tight upper bound on the variance of the surprisal, only in terms of the dimension $d$ of the system. A non-tight upper bound is implicit in \cite{polyanskiyetal}, where the term $\sum_ip_i\log^2p_i$ in (\ref{expressvariationofinformationwithpi}) has been bounded. For the expectation value of the surprisal, i.e.\ the entropy, a tight upper bound is of course well-known:\ $S(\rho)\leq\log d$.
\begin{theorem}[Maximum variance of the surprisal]\label{maxvarianceinfothm}
Let $\rho$ be a state on a $d$-dimensional system. Then, for $d\geq2$,
\begin{align}\label{firstineqformaxvarianceinfo}
\var_\rho(-\log\rho)~\leq~N(d)~<~N_d~.
\end{align}
(See definitions (\ref{definerealN}) and (\ref{defineapproxN}) for $N(d)$ and $N_d$, and cf.\ Lemma \ref{lemmareducingonevariableoptimization}.) For $d=1$, $\var_\rho(-\log\rho)=0$.

For $d\geq2$, let $r=r_d$ be the (unique) parameter attaining the maximum in the definition of $N(d)$ (Eq.\ (\ref{definerealN})). Then equality $\var_\rho(-\log\rho)=N(d)$ is achieved if and only if $\rho$ has spectrum
\begin{align}\label{formofrhowidehatinvariancethm}
{\rm spec}[\rho]~=~\left(1-r_d,\frac{r_d}{d-1},\ldots,\frac{r_d}{d-1}\right)~.
\end{align}
\end{theorem}
Theorem \ref{maxvarianceinfothm} is proved in Section \ref{proofofmaxvariancesubsect} by the method of Lagrange multipliers.

\begin{remark}[The quantities $N(d)$ and $N_d$]\label{heatcapacityboundistight}
$N(d)$ from Eq.\ (\ref{definerealN}) is well approximated by the easily computable $N_d\equiv\frac{1}{4}\log^2(d-1)+1$ since, by Lemma \ref{lemmareducingonevariableoptimization} (see also Fig.\ \ref{Mfigure}, lower right panel),
\begin{align}
N_d~>~N(d)~>~N_d-1~.
\end{align}
One can even show $N(d)=N_d-O\left(1/\log^2d\right)$ for $d\to\infty$, the optimal $r$ in (\ref{definerealN}) being $r_d=1/2-1/\log(d-1)+O\left(1/\log^2d\right)$. Instead of the maximization (\ref{definerealN}), one may compute $N(d)$ numerically by finding the optimal $r=r_d\in[0,1/2]$ as the (unique) solution of $(1-2r)\log\left(\frac{1-r}{r}(d-1)\right)=2$ and plugging it back.

Note that the quantity $N(d)$ from the optimal upper bound (\ref{firstineqformaxvarianceinfo}) appears in the quadratic Taylor term of the optimal lower bound $M(\Delta,d)$ in (\ref{lowerboundeqninmaintheorem}) as well (cf.\ Remark \ref{taylorexpansionremark}). This can be understood in a pedestrian way by minimizing $D(\rho+\varepsilon\|\rho)$ at fixed $\rho$ and for small $\varepsilon$ (with $[\rho,\varepsilon]=0$; see beginning of the proof of Theorem \ref{maintheoremineqpaper}) under the constraint $S(\rho+\varepsilon)-S(\rho)=\delta$ (small), which gives $\delta^2/2\var_\rho(-\log\rho)+O(\delta^3)$. Finally minimizing this over all $\rho$, the quadratic term of $M(\delta,d)$ is therefore $\delta^2/2N(d)$ by Theorem \ref{maxvarianceinfothm}.
\end{remark}

\subsubsection{Maximum heat capacity in finite dimensions}\label{maxheatcapacitysubsection}
We now explain the thermodynamic significance of Theorem \ref{maxvarianceinfothm} (for a more detailed exposition of the thermodynamics background see also \cite[Appendix\ A]{longlandauercompanionpaper}). Let $H$ be a \emph{Hamiltonian} of a $d$-dimensional system, i.e.\ a Hermitian $d\times d$-matrix (diagonal for classical systems); this operator determines the physical energy of the system. Then, at any \emph{temperature} $T\in(0,\infty)$, the corresponding \emph{thermal} (or equilibrium) state is
\begin{align}\label{thermalstateinvariancesubsecteqn}
\rho_T~:=~\frac{e^{-H/T}}{\tr{e^{-H/T}}}~,
\end{align}
with units chosen such that Boltzmann's constant $k_B=1$. The (average) \emph{energy} of the thermal state is the energy expectation value $E(T)~:=~\tr{H\rho_T}$, and the \emph{heat capacity} $C(T)$ quantifies the rate of change of the system energy upon temperature variation:
\begin{align}\label{computeheatcapacity}
C(T)~:=~\left.\frac{dE}{dT}\right|_T~=~\frac{d}{dT}\tr{H\frac{e^{-H/T}}{\tr{e^{-H/T}}}}~=~\var_{\rho_T}(H/T)~=~\var_{\rho_T}(-\log\rho_T)~,
\end{align}
where we omitted the little computation of the derivative, and used in the last step that the variance is unchanged under addition of a constant term (proportional to $\1$).

Eq.\ (\ref{computeheatcapacity}) shows that the heat capacity does not depend on $H$ and $T$ separately, but only on the thermal state $\rho_T$. Note that every full-rank state $\rho$ can be interpreted as the thermal state of some Hamiltonian $H_\rho:=-\log\rho$, and common extensions of the above framework include even some (or all) non-full-rank states; it is for example conventional to allow $T\in[0,\infty]$ and define $\rho_0$ to be the normalized projector onto the ground space of $H$, $H/\infty:=0$, and $C(\infty):=\lim_{T\to\infty}C(T)$.

Further note that, by (\ref{computeheatcapacity}), the heat capacity also equals the energy fluctuations $\var_{\rho_T}(H)$, i.e.\ the second moment of the energy, up to a factor of $T^2$.

Theorem \ref{maxvarianceinfothm} has thus the following corollary:
\begin{corollary}[Maximum heat capacity in $d$ dimensions]\label{maxheatcapacity}
Let $H$ be any Hamiltonian on a $d$-dimensional system, and let $T\in[0,\infty]$. Then its heat capacity $C(T)$ is uniformly bounded in terms of the dimension:\ for $d\geq2$,
\begin{align}\label{eqnoptimalupperboundonheatcapacity}
C(T)~\leq~N(d)~<~N_d~\equiv~\frac{1}{4}\log^2(d-1)+1~,
\end{align}
with $N(d)$ from Eq.\ (\ref{definerealN}). For $d=1$, $C(T)=0$.
\end{corollary}
Note that the first bound in (\ref{eqnoptimalupperboundonheatcapacity}) is tight for any $d$:\ the optimal state $\rho$ from (\ref{formofrhowidehatinvariancethm}) has full-rank and is thus the thermal state of the Hamiltonian $H:=-\log\rho$ at temperature $T:=1$.

\begin{remark}[Exponential family of optimal states (\ref{formofsigmaandrhoinDvsDeltaLemma}) and (\ref{formofrhowidehatinvariancethm})]\label{exponentialfamilyremark}The optimal states $\rho$ and $\sigma$ from (\ref{formofsigmaandrhoinDvsDeltaLemma}) come, for all values of $\Delta$, from the same exponential family:\ defining a $d$-dimensional ``Hamiltonian'' $H_{opt}:=\diag(-1,0,\ldots,0)$, it is $\sigma,\rho=e^{-H_{opt}/T_{\sigma,\rho}}/\tr{e^{-H_{opt}/T_{\sigma,\rho}}}$ for some ``temperatures'' $T_{\sigma,\rho}\in[0,\infty]$. The same is true for the state (\ref{formofrhowidehatinvariancethm}) having maximal surprisal variance or heat capacity; thermal states with one large occupation number (eigenvalue) $\approx1/2$ and completely degenerate small occupations have thus the largest energy fluctuations \cite{mackay}.

On an $N$-particle system, e.g.\ the space $\C^d=(\C^l)^{\otimes N}$ of $N$ $l$-level particles, the Hamiltonian $H_{opt}$ means physically that the system energy is minimized ($-1$) when each of the $N$ particles is in a preferred state $\ket{0}$ and equals $0$ otherwise, irrespective of the specific state. This very strong interaction between all $N$ particles leads, at some temperature $T_{crit}$, to the largest possible heat capacity of any $d=l^N$-dimensional system by Corollary \ref{maxheatcapacity},
\begin{align}\label{quadraticheatcapacity}
C(T_{crit})~=~N(l^N)~\simeq~N_{d=l^N}~\simeq~\frac{1}{4}\log^2l^N~=~N^2\,\frac{\log^2l}{4}~.
\end{align}
This is in stark contrast to a system of $N$ independent (non-interacting) particles, whose heat capacity is proportional to $N$, i.e.\ ``extensive'', whereas (\ref{quadraticheatcapacity}) is faster than extensive. Such extensivity is also usually assumed in thermodynamics e.g.\ by the Dulong-Petit law \cite{thermodynamicstheory}, at least for the most commonly considered systems made up of weakly-interacting particles.

When a system's heat capacity $C(T)/N$ per particle diverges at some temperature $T=T_{crit}$, one sometimes speaks of a second-order phase transition, and the system can then absorb or release energy density by just ``reorganizing'' its state without temperature change \cite{mackay}. Corollary \ref{maxheatcapacity} shows explicitly that such effects cannot occur for finite(-dimensional) systems.
\end{remark}

\section{Applications}\label{applicationsection}
Here we outline some implications for thermodynamics and information theory of Theorem \ref{maintheoremineqpaper}, the inequality relating relative entropy and entropy difference (see Section \ref{relentvsentdifferencesubsection}).

\subsection{Thermodynamics applications}
In Section \ref{stepwiseequlibrationsubsect} we examine how slowly equilibration processes \cite{andersgiovannetti} have to be conducted to make them (close to) \emph{thermodynamically reversible}. A relation between an intensive and an extensive quantity in many-particle systems is given in Section \ref{SdensityVsFsubsection}. Regarding the extensivity of the heat capacity in many-body systems, see also the previous Remark \ref{exponentialfamilyremark}.

The following sections also serve to illustrate the prominence of relative entropy and entropy difference in thermodynamics and statistical physics.

\subsubsection{Approach to reversibility in equilibration processes}\label{stepwiseequlibrationsubsect}
In thermodynamics it is a common assumption (which can be justified in specific models) that a system with a Hamiltonian $H$ and in weak interaction with an environment at temperature $T$ will ``equilibrate'' to the thermal final state $\rho_f=e^{-H/T}/\tr{e^{-H/T}}$ (see \cite[Appendix\ A]{longlandauercompanionpaper} and Section \ref{maxheatcapacitysubsection} above), irrespective of its initial state $\rho_i$. The system's energy change associated with such a spontaneous state change is called \emph{heat flow} or \emph{heat} $\Delta Q$ \cite{puszworonowicz,andersgiovannetti}:
\begin{align}
\Delta Q~:=~\tr{(\rho_f-\rho_i)H}~=~T\,\tr{(\rho_f-\rho_i)(-\log\rho_f)}~.
\end{align}
One can relate this to the system's entropy change $\Delta S:=\Delta(\rho_f,\rho_i)=S(\rho_f)-S(\rho_i)$:
\begin{align}\label{onestepprocessquantityeqn}
\frac{\Delta Q}{T}~=~\Delta S-D(\rho_i\|\rho_f)~\leq~\Delta S~.
\end{align}
(In order for all quantities to be well-defined, we assume $\rho_f$ to be a full-rank state, i.e.\ assume $T\in(0,\infty]$; for simplicity and without further mentioning, we assume all states in this section to be of full-rank or at least of the same support.)

The above equilibration processes can also be conducted in a stepwise fashion, which was presented and analyzed in detail by Anders and Giovannetti \cite{andersgiovannetti}. One can view this as an attempt to formalize the vague notion of ``slowness'' of an equilibration process, which according to common physics folklore should make the process ``thermodynamically reversible''. We now recapitulate some elements from \cite{andersgiovannetti} and complement their analysis by a lower bound on how ``close'' a process can be to reversibility.

In a $k$-step process, adjust the system Hamiltonian successively first to $H_1$, then instantaneously to $H_2$, \ldots, and finally to $H_k\equiv H$, and let the system equilibrate with an environment at temperature $T_j$ in each step $j=1,\ldots,k$ (often, it will be either $H_j\equiv H$ for all $j$, or $T_j\equiv T$ for all $j$). We denote the associated intermediate thermal states by $\rho_j:=e^{-H_j/T_j}/\tr{e^{-H_j/T_j}}$ (note, $\rho_k=\rho_f$) and define $\rho_0:=\rho_i$. The entropy change $\Delta S$ of the overall process equals just the sum of all changes $\Delta(\rho_j,\rho_{j-1})$, and the sum of the single-step quantities $\Delta Q_j/T_j$ satisfies, by (\ref{onestepprocessquantityeqn}),
\begin{align}\label{clausiusforsteps}
\sum_{j=1}^k\frac{\Delta Q_j}{T_j}~=~\sum_{j=1}^k\big[\Delta(\rho_j,\rho_{j-1})-D(\rho_{j-1}\|\rho_j)\big]~=~\Delta S-\sum_{j=1}^kD(\rho_{j-1}\|\rho_j)~\leq~\Delta S~.
\end{align}
The inequality between the process quantity on the LHS and $\Delta S$ is the \emph{Clausius Theorem} \cite{andersgiovannetti}, often cited to be an incarnation of the Second Law of Thermodynamics. Note that, for $T_j\equiv T$, the LHS is just proportional to the total heat flow into the system, $\sum_j\Delta Q_j$.

In the special case \cite{andersgiovannetti} where the intermediate steps $j=1,\ldots,k-1$ are chosen such that the states $\rho_j$ interpolate linearly between $\rho_i=\rho_0$ and $\rho_f=\rho_k$, i.e.
\begin{align}\label{linearinterpolation}
\rho_j~=~\left(1-\frac{j}{k}\right)\rho_i+\frac{j}{k}\,\rho_f\qquad\text{for}~j=0,\ldots,k\,,
\end{align}
then the LHS of (\ref{clausiusforsteps}) can also be bounded from below in terms of the entropy difference \cite{andersgiovannetti}:
\begin{align}
\sum_{j=1}^k\frac{\Delta Q_j}{T_j}~&=~\Delta S\,-\,\frac{D(\rho_f\|\rho_i)+D(\rho_i\|\rho_f)}{k}\,+\,\sum_{j=1}^kD(\rho_j\|\rho_{j-1})\\
&\geq~\Delta S\,-\,\frac{D(\rho_f\|\rho_i)+D(\rho_i\|\rho_f)}{k}~.\label{lowerboundonprocessquantity}
\end{align}
Thus, as the number of steps $k$ in the interpolation (\ref{linearinterpolation}) becomes finer (and if $\rho_i$, $\rho_f$ have the same support), one has $\sum_j\Delta Q_j/T_j\to\Delta S$. This is remarkable since a priori the quantity $\sum_j\Delta Q_j/T_j$ depends on the details of the process, whereas $\Delta S=\Delta(\rho_f,\rho_i)$ depends only on its initial and final state.

Any process $\rho_i\mapsto\rho_1\mapsto\ldots\mapsto\rho_f$ satisfying equality $\sum_j\Delta Q_j/T_j=\Delta S$ is called \emph{(thermodynamically) reversible}, as intuitively one expects that the reverse of such a process leads back to the original situation. This intuition can be made rigorous for the process (\ref{linearinterpolation}):\ the entropy production $\Delta S'=-\Delta S$ of the reverse process $\rho_f\mapsto\rho_{k-1}\mapsto\ldots\mapsto\rho_i$ exactly cancels $\Delta S$, and also the process quantity $\sum_j\Delta Q'_j/T'_j$ will come close to $\Delta S'\approx-\sum_j\Delta Q_j/T_j$ by reasoning analogous to (\ref{clausiusforsteps}) and (\ref{lowerboundonprocessquantity}). For constant temperatures $T_j\equiv T$, the last fact means that (almost) no heat is produced during the entire cyclic process $\rho_i\mapsto\ldots\mapsto\rho_f\mapsto\ldots\mapsto\rho_i$, i.e.\ (almost) none of the work expended to (gradually) alter the Hamiltonian \cite{puszworonowicz} is converted to heat, which physically is a less useful form of energy than work. In actual physical realizations, thermodynamic processes become irreversible when the system state $\rho(t)$ is not at all times $t$ close to the thermal state determined by the system Hamiltonian $H(t)$ and the environment temperature $T(t)$. This happens for example when the process is conducted too fast so that the system cannot fully equilibrate at each infinitesimal step.

From this reasoning, one can quantify the \emph{degree of irreversibility} of any process $\rho_i\mapsto\rho_1\mapsto\ldots\mapsto\rho_f$ by the quantity $\sum_{j=1}^kD(\rho_{j-1}\|\rho_j)$ in (\ref{clausiusforsteps}). This corresponds to the amount of work wasted at least as heat in any cyclic completion $\rho_i\mapsto\rho_1\mapsto\ldots\mapsto\rho_f=\rho_k\mapsto\rho_{k+1}\mapsto\ldots\mapsto\rho_{k+m}\equiv\rho_i$, since $\sum_{j=1}^{k+m}\Delta Q_j/T_j\leq-\sum_{j=1}^{k}D(\rho_{j-1}\|\rho_j)$ by (\ref{clausiusforsteps}). Quantitatively, denoting the minimal temperature $T_{min}:=\min_{1\leq j\leq k}T_j$, the excess heat production is at least
\begin{align}\label{wastedwork}
W_{waste}~\geq~T_{min}\,\sum_{j=1}^kD(\rho_{j-1}\|\rho_j)~,
\end{align}
which is exact if $T_j\equiv T$ for all $j$. Theorem \ref{maintheoremineqpaper} now bounds the sum in (\ref{wastedwork}) from below:
\begin{align}
\sum_{j=1}^kD(\rho_{j-1}\|\rho_j)~&\geq~\sum_{j=1}^kM\left(\Delta(\rho_{j-1},\rho_j),d\right)~=~k\,\sum_{j=1}^k\frac{1}{k}M(\Delta(\rho_{j-1},\rho_j),d)\label{firstlowerbounddegreeofirreversiblity}\\
&\geq~k\,M\left(\sum_{j=1}^k\frac{1}{k}\Delta(\rho_{j-1},\rho_j),\,d\right)~=~k\,M\left(\frac{-\Delta S}{k},\,d\right)\label{Mboundafterconvexity}\\
&\geq~\frac{1}{k}\,\frac{(\Delta S)^2}{3\log^2d}~,\label{irreversibdegreeinvk}
\end{align}
where $d<\infty$ denotes the dimension of the system, the second inequality is by convexity of the function $M$ (Theorem \ref{propertiestheorem}), and we exemplarily used the lower bound (\ref{quadraticlowerboundforallDelta}).

Achieving a degree $\varepsilon$ of reversibility by a stepwise process thus necessitates a minimum number $k=O(1/\varepsilon)$ of steps via Eq.\ (\ref{irreversibdegreeinvk}). When $k$ interpreted as the time duration of the entire process -- assuming that each equilibration step consumes roughly equal time -- then (\ref{irreversibdegreeinvk}) substantiates the folklore whereby thermodynamically reversible processes have to be conducted ``infinitely slowly''. Our estimate is thus relevant for fundamental thermodynamics and especially for small systems \cite{popescureachlimit}, as it delineates where the idealized but commonplace notion of \emph{reversible process} can apply. It also provides new heat bounds for processes out of equilibrium in the area of non-equilibrium thermodynamics \cite{lindbladbook,jarzynskiequality,jarzynskireview}.

Although the lower bound (\ref{Mboundafterconvexity}) is essentially tight in the typical thermodynamics situation where only the entropy difference $\Delta S$ between two states is known, it becomes trivial for $\Delta S=0$. In this and other cases, when in addition the initial and final states $\rho_i$, $\rho_f$ are known, one may use an estimate similar to (\ref{firstlowerbounddegreeofirreversiblity})--(\ref{irreversibdegreeinvk}) but based on Pinsker's inequality (\ref{pinskerinequality}):
\begin{align}
\sum_{j=1}^kD(\rho_{j-1}\|\rho_j)~&\geq~\sum_{j=1}^k\frac{1}{2}\|\rho_{j-1}-\rho_j\|_1^2~\geq~\frac{k}{2}\left(\sum_{j=1}^k\frac{1}{k}\|\rho_{j-1}-\rho_j\|_1\right)^2~\geq~\frac{\|\rho_i-\rho_f\|_1^2}{2k}~.
\end{align}

On the topic of stepwise processes we finally remark that the approach to reversibility $\sum_{j=1}^k\Delta Q_j/T_j\to\Delta S$ for $k\to\infty$ is \emph{not} special to the linear interpolation process (\ref{linearinterpolation}) \cite{andersgiovannetti}. Rather, for any (piecewise continuously differentiable) curve $\rho(t)$ in state space with $\rho(0)=\rho_i$, $\rho(1)=\rho_f$, a discretization at points $0=t_0<t_1<\ldots<t_k=1$ gives
\begin{align}
\sum_{j=1}^k\frac{\Delta Q_j}{T_j}~=&~\sum_{j=1}^k\tr{(-\log\rho(t_j))(\rho(t_j)-\rho(t_{j-1}))}\\
\to&~\int_{t=0}^1\tr{(-\log\rho(t))\,d\rho(t)}~=~\int_0^1dt\,\tr{-\dot{\rho}(t)\log\rho(t)}\\
=&~\int_0^1dt\,\frac{d}{dt}\tr{\rho(t)-\rho(t)\log\rho(t)}~=~S(\rho_f)-S(\rho_i)~=~\Delta S~,
\end{align}
with convergence as the discretization becomes finer, $k\to\infty$ and $\max_j|t_j-t_{j-1}|\to0$ (i.e.\ a Riemann sum). Thus, any state change $\rho_i\mapsto\rho_f$ can be made thermodynamically reversible (when $\supp[\rho_i]=\supp[\rho_f]$). For the discretized process $\rho(t)$ we do however not have a lower convergence estimate as in (\ref{lowerboundonprocessquantity}) (the upper bound from the Clausius Theorem (\ref{clausiusforsteps}) holds of course for any discretization).

\bigskip

In this section, we have considered thermalizing processes, bringing an arbitrary state $\rho_i$ to a thermal state $\rho_f$, and have measured the heat production w.r.t.\ the Hamiltonian $H$ corresonding to the \emph{final} (thermal) state \cite{andersgiovannetti}. This leads to the Clausius inequality (\ref{clausiusforsteps}).

In \cite{longlandauercompanionpaper} we use Theorem \ref{maintheoremineqpaper} in the reverse situation where an initially thermal state $\rho_i$ is used as the resource in a process leading away from equilibrium. The heat production is again measured w.r.t.\ the system's Hamiltonian, which there however is related to the \emph{initial} state and reverses the inequality (\ref{clausiusforsteps}) \cite{andersgiovannetti,longlandauercompanionpaper}. Furthermore, the paper \cite{longlandauercompanionpaper} concerns a bipartite scenario -- the Landauer process involving a system and a thermal reservoir \cite{landauer} -- where a Second Law-like statement can be formulated more properly and where the above stepwise process may be implemented by swapping the system and reservoir states.

\subsubsection{Free energy vs.\ entropy density}\label{SdensityVsFsubsection}
To further elucidate the thermodynamic meaning of the quantity $D(\rho_i\|\rho_f)$ for a thermal final state $\rho_f=e^{-H/T}/\tr{e^{-H/T}}$ (cf.\ Eq.\ (\ref{onestepprocessquantityeqn}) in Section \ref{stepwiseequlibrationsubsect}), we relate it to the \emph{work extractable at constant temperature} from the state $\rho_i$, and then examine it in a many-particle system.

For this, consider an \emph{isothermal process}, i.e.\ where the temperature $T$ remains constant and only the Hamiltonian is changed from its initial value $H_0\equiv H$ in $k$ successive steps to $H_1$, \ldots, $H_k\equiv H$, at each of which the system equilibrates as in Section \ref{stepwiseequlibrationsubsect}. The total heat flow $\Delta Q:=\sum_{j=1}^k\Delta Q_j$ during the process then satisfies the Clausius inequality $T\Delta S\geq\Delta Q$ (Eq.\ (\ref{clausiusforsteps})), so that
\begin{align}
T\,D(\rho_i\|\rho_f)~&=~-\tr{H(\rho_f-\rho_i)}+T\left[S(\rho_f)-S(\rho_i)\right]~=~F(\rho_i)-F(\rho_f)\\
&\geq~-\Delta E+\Delta Q~=~-\Delta W~,\label{inequalityinfreeenergydifference}
\end{align}
where we have defined:\ the \emph{free energy} $F(\rho):=\tr{H\rho}-TS(\rho)$ of a state $\rho$ (at temperature $T$ and for Hamiltonian $H$); the \emph{internal energy increase} $\Delta E:=\tr{H(\rho_f-\rho_i)}$; and the \emph{work} $\Delta W:=\Delta E-\Delta Q$ done on the system \cite{puszworonowicz,procaccialevine,andersgiovannetti}.

According to Section \ref{stepwiseequlibrationsubsect}, equality in (\ref{inequalityinfreeenergydifference}) can be approached by a suitable (reversible) process (note that then the jump from $H=H_0$ to the first equilibration step $H_1\approx-T\log\rho_i$ may be big, whereas the further steps $H_1\mapsto\ldots\mapsto H_k=H$ are small). Thus, the amount $T\,D(\rho_i\|\rho_f)=(-\Delta W)_{max}$ of work \emph{can} be extracted from the state $\rho_i$ by a thermodynamic process at temperature $T$ and using the internal energy function $H$. Conversely, for given temperature $T$ and Hamiltonian $H$, this is also the \emph{maximum} amount of work extractable from $\rho_i$ since, for any process leading to a final state $\rho'_f$ (not necessarily thermal for either $H$ or $T$),
\begin{align}
-\Delta W'~&=~-\Delta E'+\Delta Q'~\leq~-\tr{H(\rho'_f-\rho_i)}+T\left[S(\rho'_f)-S(\rho_i)\right]\nonumber\\
&=~F(\rho_i)-F(\rho'_f)~=~\left[F(\rho_i)-F(\rho_f)\right]-\left[F(\rho'_f)-F(\rho_f)\right]\nonumber\\
&=~T\,D(\rho_i\|\rho_f)-T\,D(\rho'_f\|\rho_f)~\leq~T\,D(\rho_i\|\rho_f)~.\label{amountstothermodyninequ}
\end{align}
The last inequality is due to the nonnegativity of the relative entropy, which here in more physical terms amounts to the fact that the free energy $F(\rho'_f)$ attains its minimum at the thermal state $\rho'_f=\rho_f$ (uniquely for $T\neq0$).

Theorems \ref{maintheoremineqpaper} and \ref{propertiestheorem} provide thus lower bounds on the extractable work at constant temperature $T$:
\begin{align}\label{extworkentropydensity}
W_{extr,T}~=~T\,D(\rho_i\|\rho_f)~\geq~T\,M(-\Delta S,d)~\gtrsim~2T\left(\frac{\Delta S}{\log d}\right)^2~=~2T\left(\frac{\Delta s}{\log l}\right)^2~,
\end{align}
where in the last step we have assumed a system of $N$ $l$-level particles, i.e.\ $d=l^N$ (cf.\ Remark \ref{exponentialfamilyremark}), and defined the change in \emph{entropy density} $\Delta s:=\Delta S/N=(S(\rho_f)-S(\rho_i))/N$ \cite{brattelirobinson}.

Inequality (\ref{extworkentropydensity}) seems quite unusual as its LHS is the ``extensive'' free energy difference or extractable work whereas the RHS is an ``intensive'' quantity, given by the entropy density and temperature; moreover, in the ``thermodynamic limit'' $N\to\infty$, the inequality is essentially tight. The reason for this is that states attaining equality are of the form (\ref{formofsigmaandrhoinDvsDeltaLemma}), which are strongly correlated as discussed in Remark \ref{exponentialfamilyremark}, such that one cannot speak of few-particle properties and the designation ``extensive'' is not appropriate.

\subsection{Information-theoretic applications}
We have already outlined in Section \ref{subsectionvarsurprisal} the meaning of $\sqrt{\var_\rho(\log\rho)}$ as the fluctuation in codeword length of an optimal prefix code; for a source with $d$ distinct signals this fluctuation is at most $\simeq\frac{1}{2}\log d$ by Theorem \ref{maxvarianceinfothm}. In the following, we discuss implications in information theory of the lower bound on the relative entropy (Theorems \ref{maintheoremineqpaper} and \ref{propertiestheorem}).

\subsubsection{Cost of wrong code, universal codes, and Shannon channel capacity}\label{wrongcodesubsubsection}
For a source producing i.i.d.\ signals $i$ according to a classical probability distribution $\rho=\{p_i\}_{i=1}^d$, Shannon's source compression theorem \cite{shannonmathemtheorycommun,coverthomas} shows any prefix code with a $D$-ary alphabet to have an average length of at least $S(\rho)/\log D$ per encoded signal. The lower the entropy of the signal distribution, the shorter on average the encoded message can be. This length is in fact achievable -- up to less than $1$ alphabet symbol -- by assigning codewords of length $\lceil-\log_Dp_i\rceil$ to the signals.

If one however wrongly assumes the signals $i$ to be distributed according to $\sigma=\{q_i\}_{i=1}^d$ and constructs a code for this distribution, with codewords of length $\lceil-\log_Dq_i\rceil$, then the average code length $L_\sigma$ will be
\begin{align}
L_\sigma~=~\sum_{i=1}^dp_i\left\lceil-\log_Dq_i\right\rceil~\geq~-\frac{1}{\log D}\sum_{i=1}^dp_i\log q_i~=~\frac{S(\rho)}{\log D}+\frac{D(\rho\|\sigma)}{\log D}~.
\end{align}
The last term is the cost of the wrong code \cite{coverthomas} beyond the optimal average code length $S(\rho)/\log D$ when one knew the correct distribution.

Theorems \ref{maintheoremineqpaper} and \ref{propertiestheorem} give a lower bound on this penalty just in terms of the difference $\delta=\left(S(\rho)-S(\sigma)\right)/\log D$ between the supposed length $S(\sigma)/\log D$ and the optimal achievable length $S(\rho)/\log D$:
\begin{align}\label{lowerboundoncostofwrongcode}
\frac{D(\rho\|\sigma)}{\log D}~\geq~\frac{M(\delta\log D,d)}{\log D}~\geq~\delta^2\frac{2\log D}{\log^2(d-1)+4}~,
\end{align}
where the last inequality holds only for positive expected savings $\delta\geq0$.

\bigskip

When the signals $i\in\{1,\ldots,d\}$ follow one of the distributions $\rho^\theta$, where the parameter $\theta\in\{1,\ldots,m\}$ is not known, one may choose a coding distribution $\sigma$ (``universal code'') that minimizes the maximal occuring penalty or \emph{redundancy} (see \cite{coverthomas}, Section 13.1):
\begin{align}\label{minimalredundancyeqn}
R^*~:=~\min_\sigma\max_\theta D(\rho^\theta\|\sigma)~.
\end{align}
The theorems from Section \ref{relentvsentdifferencesubsection} give an easily computable lower bound on the quantity $R^*$:\ for this, denote by $S_{min}$ and $S_{max}$ the minimal resp.\ maximal entropy $S(\rho^\theta)$ among the states $\rho^\theta$. Then using the properties from Theorem \ref{propertiestheorem} and Remark \ref{Mincreasingremark}, we have:
\begin{align}
R^*~&\geq~\min_\sigma\max_\theta\,M(S(\rho^\theta)-S(\sigma),d)~=~\min_{S\in[0,\log d]}\max_\theta\,M(S(\rho^\theta)-S,d)\label{asdfqwerasdfyxcv}\\
&=~\min_{S\in[S_{min},S_{max}]}\max\big\{M\left(S_{max}-S,d\right),\,M\left(S_{min}-S,d\right)\,\big\}\label{lastlinementioningMasdfasdf}\\
&\geq~\min_{S\in[S_{min},S_{max}]}~\max_{\Delta\in\{S_{max}-S,\,S_{min}-S\}}\,\left[\frac{2\Delta^2}{\log^2(d-1)+4}+\frac{8\Delta^3}{3\left(\log^2(d-1)+4\right)^2}\right]\\
&\geq~\frac{(S_{max}-S_{min})^2}{2\left(\log^2(d-1)+4\right)}\,-\,\frac{(S_{max}-S_{min})^3}{3\left(\log^2(d-1)+4\right)^2}~,\label{afterassumingsmallerfordeltaleq123}
\end{align}
where the last inequality follows from the observation that for every $S$ one has $|S_{max}-S|\geq(S_{max}-S_{min})/2$ or $|S_{min}-S|\geq(S_{max}-S_{min})/2$. If the following conjecture holds, a stronger lower bound $R^*\geq M(-(S_{max}-S_{min})/2,d)$ would follow from line (\ref{lastlinementioningMasdfasdf}) by the same reasoning.
\begin{conjecture}\label{conjecture}For any $d\geq2$ and any $\Delta\in[0,\log d]$, it is $M(\Delta,d)\,\geq\,M(-\Delta,d)$.\end{conjecture}
While there is numerical evidence in favor of this conjecture, e.g.\ by plotting for many values of $d\geq2$ the functions $(M(\Delta,d)-M(-\Delta,d))$ in the range $\Delta\in[0,\log d]$ and observing their nonnegativity, and while the conjecture is consistent with all previous analytical results (e.g.\ Theorem \ref{propertiestheorem} and first paragraph of Remark \ref{taylorexpansionremark}), we have not been able to prove it. In particular, the stronger conjecture that $D_2(s\|r)\geq D_2(r\|s)$ holds for the optimal pair $(s,r)$ in the minimization (\ref{definefunctionM}) of $M(\Delta,d)$ at any positive $\Delta>0$ is generally wrong; e.g.\ for $d=1000$ and $\Delta=6$, it is $s\approx0.9497$, $r\approx0.0723$ and thus $D_2(s\|r)\approx2.30<2.51\approx D_2(r\|s)$.

\bigskip

The minimal redundancy $R^*$ from (\ref{minimalredundancyeqn}) equals the Shannon capacity $C(T)$ (measured in $\text{nats}$) of the classical discrete memoryless channel $T:\theta\mapsto i$ that is defined by the transition probabilities $T(i|\theta):=p^\theta_i$ (Theorem 13.1.1 in \cite{coverthomas}, originally due to \cite{gallager,ryabko}; the proof uses a minimax theorem). This gives as above:
\begin{proposition}[Lower bound on the classical Shannon capacity]\label{capacityproposition}For a discrete memoryless channel $T:{\mathcal X}\to{\mathcal Y}$, given by transition probabilities $T(y|x)$ and with finite output dimension $|{\mathcal Y}|\geq2$, the Shannon capacity $C(T)$ is bounded from below as
\begin{align}\label{lowerboundshannoncapacityeqn}
C(T)~\geq~\frac{(S_{max}-S_{min})^2}{2\left(\log^2(|{\mathcal Y}|-1)+4\right)}\,-\,\frac{(S_{max}-S_{min})^3}{3\left(\log^2(|{\mathcal Y}|-1)+4\right)^2}~,
\end{align}
where $S_{max}$ and $S_{min}$ denote the maximal and minimal entropies, respectively, of any column $T(\cdot|x)$ of the transition matrix.
\end{proposition}

Again, if Conjecture \ref{conjecture} holds we would have the stronger bound $C(T)\geq M((S_{min}-S_{max})/2,|{\mathcal Y}|)$. This bound or the bound from Proposition \ref{capacityproposition} are easier to evaluate than Shannon's mutual information formula for the exact $C(T)$ \cite{shannonmathemtheorycommun,coverthomas}. When one bounds the relative entropies in (\ref{minimalredundancyeqn}) from below by Pinsker's or the H-O-T inequality (\ref{pinskerinequality}) \cite{pinskersinequality,HOT81,refinementspinskier,audenaerteisertoptimal}, one would obtain a linear program in the variables $\sigma$. Also, Proposition \ref{capacityproposition} provides a more systematic way to obtain lower bounds on $C(T)$ than by plugging trial input distributions into Shannon's formula. 

On the other hand, the lower bound (\ref{lowerboundshannoncapacityeqn}) will be trivial iff all columns $T(\cdot|x)$ have the same entropy, whereas the capacity $C(T)$ vanishes only iff all columns are themselves identical. Also, the lower bound in (\ref{lowerboundshannoncapacityeqn}) can never exceed $(\log2)=1\,\text{bit}$, since it has to hold for input dimension $|{\mathcal X}|=2$ as well (or when there are only two distinct columns in $T(\cdot|x)$); in the most favorable case $S_{max}-S_{min}=\log d$, the RHS of (\ref{lowerboundshannoncapacityeqn}) is actually always between $0.111\simeq0.16\,\text{bit}$ (for $d=2$) and $\log\sqrt{3}\simeq0.80\,\text{bit}$ (for $d\to\infty$; cf.\ Remark \ref{taylorexpansionremark}).

\bigskip

In the quantum setting, identical formulas apply for the cost of the wrong code (\ref{lowerboundoncostofwrongcode}) and the redundancy (\ref{afterassumingsmallerfordeltaleq123}), see \cite{schumacherwestmorelandreview,indeterminatelengthQcoding}. Furthermore, the Holevo quantity, which is a lower bound on the classical capacity of a quantum channel \cite{nielsenchuang}, equals the relative entropy radius of the channel output, i.e.\ the redundancy (\ref{minimalredundancyeqn}) over all output states \cite{OPW97,schumacherwestmorelandreview}. For a quantum channel, however, there is no systematic way known in particular to find the \emph{minimum output entropy} $S_{min}$ efficiently; the channel output set has, e.g., generally infinitely many extreme points.

\subsubsection{Hypothesis testing and large deviations}
The relative entropy features prominently also in hypothesis testing and large deviation theory \cite{coverthomas}. On the one hand, relative entropies $D(\sigma\|\rho)$ between given states $\sigma$, $\rho$ appear for example as error exponents in asymmetric hypothesis testing (in the classical Chernoff-Stein Lemma \cite{coverthomas} as well as in its quantum analogue \cite{hiaipetz91,ogawanagaoka2000}), such that Theorems \ref{maintheoremineqpaper} and \ref{propertiestheorem} apply immediately to yield lower bounds on error decay rate in terms of the entropy difference $S(\sigma)-S(\rho)$ only.

On the other hand, in these areas one is often interested in quantities like
\begin{align}\label{lowerbounddistance}
{\rm{dist}}(E,\rho)~:=~\inf_{\sigma\in E}D(\sigma\|\rho)~,
\end{align}
where $E$ is some set of $d$-dimensional probability distributions and $\rho$ a fixed distribution. Sometimes the set $E$ is described by an entropy constraint, for example in universal coding for all $d$-dimensional sources of entropy less than $R$ (\cite{coverthomas}; similarly \cite{bjelakovic} for the quantum case):\ here, the decoding error probability vanishes exponentially in the message length $n$ like $\sim\exp(-n\,{\rm{dist}}(E,\rho))$ if the true source distribution is $\rho$ (assuming $S(\rho)<R$) and where $E:=\{\sigma|S(\sigma)>R\}$. The decay rate, ${\rm{dist}}(E,\rho)$, may thus be bounded from below by $M(R-S(\rho),d)$ according to Theorems \ref{maintheoremineqpaper} and \ref{propertiestheorem} simply in terms of an entropy difference.

Finally, in symmetric hypothesis testing between two classical (commuting) probability distributions $\rho_1$, $\rho_2$, the optimal error decay rate is given by the \emph{Chernoff information} $\xi(\rho_1,\rho_2)=-\log\min_{0\leq s\leq1}\tr{\rho_1^s\rho_2^{1-s}}$ \cite{classchernoff,coverthomas}, which has the property that there exists a distribution $\sigma$ (from the Hellinger arc between $\rho_1$ and $\rho_2$) satisfying $\xi(\rho_1,\rho_2)=D(\sigma\|\rho_1)=D(\sigma\|\rho_2)$. Similar to the derivation leading up to (\ref{afterassumingsmallerfordeltaleq123}), the latter quantity can be bounded from below in terms of the entropy difference $\Delta(\rho_1,\rho_2)=S(\rho_1)-S(\rho_2)$ between the two states only:
\begin{align}\label{lowerboundonchernoffinfoeqn}
\xi(\rho_1,\rho_2)~&\geq~\frac{|\Delta(\rho_1,\rho_2)|^2}{2\left(\log^2(d-1)+4\right)}\,-\,\frac{|\Delta(\rho_1,\rho_2)|^3}{3\left(\log^2(d-1)+4\right)^2}~,
\end{align}
where the last expression does not involve any extremization (cf.\ Theorem \ref{propertiestheorem}), and a better bound would follow from Conjecture \ref{conjecture} as above.

Whereas for symmetric hypothesis testing between (non-commuting) quantum states $\rho_1$, $\rho_2$ the basic formula for the decay rate $\xi(\rho_1,\rho_2)$ holds as well, the existence of a state $\sigma$ as above is not known \cite{quantumchernoff}. We can therefore not apply the same reasoning to get a lower bound on $\xi(\rho_1,\rho_2)$ in the quantum setting. For other kinds of (dimension-independent) bounds on the quantum and classical Chernoff information, see \cite{quantumchernoff,quantumsandwichbounds}.

\subsubsection{Mutual information}\label{mutualinformationsubsection}
Let $\rho_{AB}$ be a joint state on a bipartite system $AB$ with respective local dimensions $d_A$ and $d_B$ and total dimension $d=d_Ad_B$ (in the classical probabilistic case, $\rho_{AB}$ is a joint probability distribution of two random variables $A$ and $B$ with $d_A$ and $d_B$ outcomes, respectively). Then its \emph{mutual information} $I(A:B):=S(\rho_A)+S(\rho_A)-S(\rho_{AB})$ can be written as both a relative entropy and an entropy difference \cite{ohyapetz}:
\begin{align}
I(A:B)~=&~S(\rho_A\otimes\rho_B)-S(\rho_{AB})~=~-\Delta(\rho_{AB},\rho_A\otimes\rho_B)\\
=&~D(\rho_{AB}\|\rho_A\otimes\rho_B)~,
\end{align}
where $\rho_A$, $\rho_B$ denote the reduced states (marginal probability distributions) for $A$ and $B$, and in the first line we used the notation (\ref{definenotationDelta}).

Here we just remark that Theorem \ref{maintheoremineqpaper}, which relates relative entropy and entropy difference, does not give any constraints in this situation:\ for $\Delta\in[-\log d,0]$, which is the case here, it is $-\Delta\geq M(\Delta,d)$ by Remark \ref{Mincreasingremark}, with strict inequality except for $\Delta=-\log d,0$; the fact $D(\rho_{AB}\|\rho_A\otimes\rho_B)=-\Delta(\rho_{AB}\|\rho_A\otimes\rho_B)$ is thus consistent with Theorem \ref{maintheoremineqpaper} and therefore (\ref{ineqalityinmaintheorem}) does not give new information.

Note that $I(A:B)\leq\min\{\log d_A,\log d_B\}$ in the classical case, whereas for quantum states $I(A:B)\leq2\min\{\log d_A,\log d_B\}$, so that the maximum value $\log d=\log d_A+\log d_B$ of $-\Delta(\rho_{AB},\rho_A\otimes\rho_B)$ and of $D(\rho_{AB}\|\rho_A\otimes\rho_B)$ can be attained only in the quantum case and only when $d_A=d_B$ with a maximally entangled state $\rho_{AB}$ \cite{nielsenchuang}.

\section{Proofs}\label{proofsection}

\subsection{Proof of Theorem \ref{maintheoremineqpaper}}\label{proofofmainthmsubsect}

\begin{proof}[Proof of Theorem \ref{maintheoremineqpaper}]
To prove the inequality (\ref{ineqalityinmaintheorem}) and the optimality statement around (\ref{formofsigmaandrhoinDvsDeltaLemma}), we will compute, for any fixed $\Delta\in[-\log d,\log d]$, the infimum
\begin{align}\label{stateinfimum}
\inf_{\sigma,\rho}\left\{\,D(\sigma\|\rho)\,\,\big|\,\,S(\sigma)-S(\rho)=\Delta\,\right\}
\end{align}
over $d$-dimensional quantum states $\sigma$, $\rho$, and show that it equals $M(\Delta,d)$ from Eq.\ (\ref{definefunctionM}) with optimal states $\sigma$, $\rho$ of the form (\ref{formofsigmaandrhoinDvsDeltaLemma}).

We first make some basic observations about the infimum (\ref{stateinfimum}) including the fact that it is always attained. For $\Delta=\log d$, one necessarily has $\sigma=\1/d$ and $\rho$ is a pure state, so that $D(\sigma\|\rho)=\infty$ is attained; on the other hand, this equals $M(\log d,d)=\infty$, as $\Delta=\log d$ in (\ref{definefunctionM}) enforces $s=(d-1)/d$ and $r=0$; the case $\Delta=\log d$ is thus done and we exclude it from all further considerations. For any fixed $\Delta\in[-\log d,\log d)$ there exists a full-rank state $\rho$ with $S(\sigma)-S(\rho)=\Delta$, such that the infimum (\ref{stateinfimum}) is finite. As the set of pairs $(\sigma,\rho)$ satisfying $S(\sigma)-S(\rho)=\Delta$ is compact and the function $(\sigma,\rho)\mapsto D(\sigma\|\rho)$ is lower semicontinuous, the infimum is attained. For similar reasons, the infimum in (\ref{definefunctionM}) is attained. For the argumentation below, we note further that $\binH(s)+s\log(d-1)$ is strictly increasing in $s\in[0,(d-1)/d]$ from the value $0$ at $s=0$ to $\log d$ at $s=(d-1)/d$ with first derivative
\begin{align}\label{paradigmaticallydifferentiateG2}
\frac{d}{ds}\left(\binH(s)+s\log(d-1)\right)~=~\log\left(\frac{1-s}{s}(d-1)\right)\qquad\text{for}~\,s\in(0,1)~.
\end{align}

 \bigskip

It is easy to see that the infimum in (\ref{stateinfimum}) is attained for \emph{commuting} states $\sigma$ and $\rho$:\ fixing the state $\rho$ and fixing all eigenvalues ${\rm spec}(\sigma)$ of $\sigma$ (which also fixes the entropy $S(\sigma)$; this should be done to be consistent with $S(\sigma)-S(\rho)=\Delta$ for the fixed $\Delta$), the infimum (over $\sigma$) of the relative entropy
\begin{align}\label{writerelentropyasenergy}
D(\sigma\|\rho)~=~-S(\sigma)\,+\,\tr{(-\log \rho)\sigma}
\end{align}
is attained by the state $\sigma$ which is diagonal in the same basis as $(-\log\rho)$ and has its eigenvalues ordered in the opposite way as $(-\log\rho)$ \cite{bhatiabook}; as the logarithm is a strictly increasing function, $\sigma$ will thus also be diagonal in the same basis as $\rho$ (and in particular commute with $\rho$), with its eigenvalues ordered in the same way as $\rho$. (When ${\rm rank}(\rho)<{\rm rank}({\rm spec}[\sigma])$, the infimum is $+\infty$, and this as well can be attained by a $\sigma$ commuting with $\rho$). This commutativity carries over to the infimum in (\ref{stateinfimum}), and implies that the bound we are about to prove will be optimal for the case of classical $d$-dimensional probability distributions (i.e.\ diagonal density matrices) as well.

\bigskip

One can get more information about the optimal pair $(\sigma,\rho)$ from Klein's inequality, i.e.\ the nonnegativity of the relative entropy. We fix again the state $\rho$ and fix the entropy of $\sigma$ to equal $S(\sigma)=S$, leaving the spectrum of $\sigma$ otherwise free; under these constraints we again minimize (\ref{writerelentropyasenergy}). In thermodynamics language (see Eq.\ (\ref{thermalstateinvariancesubsecteqn}) and below), this is the minimization of the ``energy'' of $\sigma$ w.r.t.~the ``Hamiltonian'' $(-\log \rho)$ under the entropy constraint $S(\sigma)=S$; by the ``thermodynamic inequality'' (e.g.\ \cite{ohyapetz}), a version of Klein's inequality, it is well-known that the minimum is attained for a ``thermal state'' $\sigma\sim e^{-\gamma(-\log\rho)}$, i.e.~$\sigma=\rho^\gamma/\tr{\rho^\gamma}$, for some ``inverse temperature'' $\gamma\in[0,+\infty]$ (here we define $0^0:=0$, and $\rho^\infty/\tr{\rho^\infty}$ is to be understood as the maximally mixed state on the eigenspace of $\rho$ corresponding to its largest eigenvalue).

Making this argument more precise requires some care. We consider the minimization of (\ref{writerelentropyasenergy}) under variation of \emph{both} $\sigma$ and $\rho$ with the constraints of fixed $S(\rho)=S_\rho$ and fixed $S(\sigma)=S_\sigma=S_\rho+\Delta$, and denote by $(\widehat{\sigma},\widehat{\rho})$ a minimizing assignment. Only in the case $S_\rho=0$ we can have $D(\widehat{\sigma}\|\widehat{\rho})=+\infty$, and we do not consider this case here as it is only necessary if $\Delta=\log d$, which was already discussed above. Thus $D(\widehat{\sigma}\|\widehat{\rho})<\infty$, and so we have $\supp[\widehat{\sigma}]\subseteq\supp[\widehat{\rho}]$, which implies $\log{\rm rank}(\widehat{\rho})\geq S_\sigma$. Now, if $S_\sigma=\log{\rm rank}(\widehat{\rho})$, then obviously $\widehat{\sigma}=\widehat{\rho}^0/\tr{\widehat{\rho}^0}$ (i.e.~$\widehat{\sigma}$ is the maximally mixed state on the support of $\widehat{\rho}$; we define $0^0:=0$). Second, if $\log{\rm rank}(\widehat{\rho})> S_\sigma>\log m_0$, where $m_0$ denotes the dimension of the eigenspace of the largest eigenvalue of $\widehat{\rho}$ (i.e.~the dimension of the ground state space of the ``Hamiltonian'' $(-\log\widehat{\rho})$), then due to continuity of the entropy \cite{fannesinequalitypaper,audenaertfannes} there exists $\gamma\in(0,\infty)$ with $S(\widehat{\rho}^\gamma/\tr{\widehat{\rho}^\gamma})=S_\sigma$. We claim that then $\widehat{\sigma}=\widehat{\rho}^\gamma/\tr{\widehat{\rho}^\gamma}$ is the unique minimizer of (\ref{writerelentropyasenergy}) under variation of $\sigma$ (when keeping $\rho=\widehat{\rho}$ fixed). This is easy to see by verifying $\gamma\left(D(\sigma\|\widehat{\rho})-D(\widehat{\rho}^\gamma/\tr{\widehat{\rho}^\gamma}\|\widehat{\rho})\right)=D(\sigma\|\widehat{\rho}^\gamma/\tr{\widehat{\rho}^\gamma})$ for all states $\sigma$ with $S(\sigma)=S_\sigma$, and then using that $D(\sigma\|\widehat{\rho}^\gamma/\tr{\widehat{\rho}^\gamma})\geq0$ with equality iff $\sigma=\widehat{\rho}^\gamma/\tr{\widehat{\rho}^\gamma}$ (by Klein's inequality). Third, if $S_\sigma=\log m_0$, then the maximally mixed state on the eigenspace of the largest eigenvalue of $\widehat{\rho}$ is obviously the unique state with entropy $S_\sigma$ and minimizing (\ref{writerelentropyasenergy}), i.e.~we could formally write $\widehat{\sigma}=\widehat{\rho}^\infty/\tr{\widehat{\rho}^\infty}$. Fourth, if $S_\sigma<\log m_0$, then $\widehat{\sigma}$ may be any state supported on the eigenspace of the largest eigenvalue of $\widehat{\rho}$. In all of these case, $\widehat{\sigma}$ and $\widehat{\rho}$ commute, which was already seen above by simpler reasoning.

For the following we will thus write the minimizing assignment $(\widehat{\sigma},\widehat{\rho})$ from the previous paragraph as follows:
\begin{align}\label{commutingsigmanadrho}
\widehat{\sigma}~=~{\rm diag}(\widehat{q}_1,\ldots,\widehat{q}_d)\qquad\text{and}\qquad\widehat{\rho}~=~{\rm diag}(\widehat{p}_1,\ldots,\widehat{p}_d)~.
\end{align}
Now fixing $\sigma=\widehat{\sigma}$ in (\ref{writerelentropyasenergy}), the minimization over all commuting states $\rho={\rm diag}(p_1,\ldots,p_d)$ leads to the Lagrange function
\begin{align}\label{writelagrangefunction}
L(\{p_i\},\nu,\mu)~:=~\sum_i\left(\widehat{q}_i\log\widehat{q}_i-\widehat{q}_i\log p_i\right)\,+\,\nu\sum_ip_i\,+\,\mu\sum_ip_i\log p_i
\end{align}
with Lagrange multipliers $\nu$ and $\mu$ corresponding to the nomalization and entropy constraints $\tr{\rho}=1$ and $S(\rho)=S_\rho$, respectively. We now look at this as a function of those variables $p_i$, for which the corresponding element $\widehat{p}_i\neq0$ is positive (i.e.~which lie in the interior of the domain of $L$), and we fix the other elements $p_i$ to be zero. Then, since $p_i=\widehat{p}_i$ is a minimizing assignment, by the method of Lagrange multipliers we are guaranteed one of the following two things:\ either there exist $\widehat{\nu},\,\widehat{\mu}\in(-\infty,+\infty)$ such that
\begin{align}\label{lagrangeequationdifferentiated}
\left.\frac{dL}{dp_j}\right|_{\{\widehat{p}_i\},\widehat{\nu},\widehat{\mu}}~=~-\frac{\widehat{q}_j}{\widehat{p}_j}\,+\,(\widehat{\nu}+\widehat{\mu})\,+\,\widehat{\mu}\log\widehat{p}_j~=~0\qquad\forall j~\,\text{with}~\,\widehat{p}_j\neq0~;
\end{align}
or the gradients of the constraints $\tr{\rho}$ and $S(\rho)$ in (\ref{writelagrangefunction}) are linearly dependent at $\rho=\widehat{\rho}$, i.e.
\begin{align}\label{linearlydependentgradients}
\Bigg\{{\rm grad}_{\{j:\,\widehat{p}_j\neq0\}}\Big(\sum_ip_i\Big)\Bigg|_{\{\widehat{p}_i\}},~{\rm grad}_{\{j:\,\widehat{p}_j\neq0\}}\Big(\sum_ip_i\log p_i\Big)\Bigg|_{\{\widehat{p}_i\}}\Bigg\}\quad\text{is linearly dependent}.
\end{align}
We will now examine (potential) minimizing assignments satisfying (\ref{lagrangeequationdifferentiated}), and at the end of this proof we will show that the solutions of (\ref{linearlydependentgradients}) yield no (new) minimizers.

Eq.\ (\ref{lagrangeequationdifferentiated}) does not allow the fourth case from the paragraph before Eq.\ (\ref{commutingsigmanadrho}) as a minimizing assignment, since in this case there are $\widehat{p}_j=\widehat{p}_k=\lambda_{max}(\widehat{\rho})>0$ and $\widehat{q}_j\neq\widehat{q}_k$, contradicting (\ref{lagrangeequationdifferentiated}). Within the third case of the same paragraph, it excludes the possibility that, apart from the maximum eigenvalue $\lambda_{max}(\widehat{\rho})$, there could be two further distinct non-zero eigenvalues $\widehat{p}_i\neq\widehat{p}_j$, as in the third case both of these would have corresponding $\widehat{q}_i=\widehat{q}_j=0$, again contradicting (\ref{lagrangeequationdifferentiated}). Thus, in the third case above, $\widehat{\rho}$ has at most two distinct non-zero eigenvalues, as does $\widehat{\sigma}$.

Also for the first and second cases in the paragraph before Eq.\ (\ref{commutingsigmanadrho}) we now want to show that, except possibly when $\gamma=1$ (i.e.\ for $\widehat{\sigma}=\widehat{\rho}$ or $\Delta=0$), Eq.\ (\ref{lagrangeequationdifferentiated}) allows $\widehat{\rho}$ to have at most two distinct non-zero eigenvalues, and $\widehat{\sigma}$ as well. In these two cases, we have $\widehat{q}_j=\widehat{p}_j^\gamma/Z$ for some $\gamma\in[0,\infty)$ with $Z:=\sum_i\widehat{p}_i^\gamma>0$. Now define $x_j:=Z\widehat{q}_j/\widehat{p}_j=\widehat{p}_j^{\gamma-1}$ for each $j$ with $\widehat{p}_j>0$. Eq.~(\ref{lagrangeequationdifferentiated}) says then that, for $\gamma\neq1$, the points $x_j$ lie at intersections of the non-horizontal affine function $-x/Z+(\widehat{\nu}+\widehat{\mu})$ with the function $-(\widehat{\mu}/(\gamma-1))\log x$ (both are functions of $x>0$). The latter function is either strictly convex or strictly concave or constant (depending on whether the prefactor is negative or positive or zero). The two functions can thus not intersect at more than 2 distinct points $x_j>0$. When $\gamma\neq1$, there can therefore be at most 2 distinct non-zero values of $x_j$, i.e.\ also at most 2 distinct non-zero values of $\widehat{p}_j$ and of $\widehat{q}_j$.

\bigskip

Summing up so far, any states $\rho$ and $\sigma$ attaining the infimum in (\ref{stateinfimum}) commute and, for $\Delta\neq0$ and when they satisfy Eq.\ (\ref{lagrangeequationdifferentiated}), have at most two distinct non-zero eigenvalues each, in such a way that distinct eigenvalues in $\sigma$ and in $\rho$ correspond to each other. More precisely,
\begin{align}\label{diagstateswith2entries}
\begin{split}
\sigma~=&~{\rm diag}\left(\frac{1-s}{m},\ldots,\frac{1-s}{m},\frac{s}{n},\ldots,\frac{s}{n},0,\ldots,0\right)~,\\
\rho~=&~{\rm diag}\left(\frac{1-r}{m},\ldots,\frac{1-r}{m},\frac{r}{n},\ldots,\frac{r}{n},0,\ldots,0\right)~,
\end{split}
\end{align}
where $m,n\geq1$, $m+n\leq d$ and $s,r\in[0,1]$. Permuting the entries of both states simultaneously, we may assume the entries of $\sigma$ to be ordered non-increasingly, i.e.~$(1-s)/m\geq s/n$. The above analysis showed further that the diagonal entries of a minimizing pair are ordered in the same order (see below Eq.\ (\ref{writerelentropyasenergy}); this can also be seen by the fact that the inverse temperature $\gamma$ above turned out to be always non-negative). Thus, $(1-r)/m\geq r/n$ as well, and we will therefore in the following always assume $0\leq s,r\leq n/(m+n)$. Even in the case $\Delta=0$, some of the minimizing pairs $(\sigma,\rho)$ have this form (choose any $m$, $n$, and $s=r$), and we thus assume this form below; similarly for the case $\Delta=\log d$, which is achieved by $m=1$, $n=d-1$, $s=(d-1)/d$, $r=0$. We can thus continue the optimization in (\ref{stateinfimum}) with states of the form (\ref{diagstateswith2entries}). Before that, note for the states in (\ref{diagstateswith2entries}):
\begin{align}
&S(\sigma)~=~\binH(s)+(1-s)\log m+s\log n~,\quad S(\rho)~=~\binH(r)+(1-r)\log m +r\log n~,\label{entropyforstatesofspecialform}\\
&\Delta(\sigma,\rho)~=~S(\sigma)-S(\rho)~=~\binH(s)-\binH(r)+(s-r)\log\frac{n}{m}~,\label{entropydifferenceforstateswitheroes}\\
&D(\sigma\|\rho)~=~\binrel(s\|r)~=~s\log\frac{s}{r}+(1-s)\log\frac{1-s}{1-r}~.\label{relativeentropyforstatesofspecialform}
\end{align}

Given $\Delta\neq0$, let now the states $\sigma$ and $\rho$ in (\ref{diagstateswith2entries}), parametrized by $s$, $r$, $m$, and $n$, attain the infimum in (\ref{stateinfimum}). Our next goal is to show $m=1$ and $n=d-1$. For now, we will denote by $\tau_{t,m,n}$ the state parametrized by $t$, $m$, and $n$, such that, for example, $\tau_{s,m,n}=\sigma$ and $\tau_{r,m,n}=\rho$ in (\ref{diagstateswith2entries}). Assume that there exist $m',n'\geq1$ with $m'+n'\leq d$ and $n'/m'>n/m$. We will then show that there exists some $s'$ such that the pair of states $(\tau_{s',m',n'},\tau_{r,m',n'})$ would achieve a strictly lower value in (\ref{stateinfimum}) than the pair $(\sigma,\rho)$. For this, compute
\begin{align}
S(\tau_{s,m',n'})-S(\tau_{r,m',n'})~=~\binH(s)-\binH(r)+(s-r)\log\frac{n'}{m'}~=~\Delta+(s-r)\log\frac{n'/m'}{n/m}~,\label{expressiontoshowdminus1}
\end{align}
and note the the last logarithm is positive due to $n'/m'>n/m$. Now, assume first $\Delta>0$. Then, from (\ref{entropydifferenceforstateswitheroes}), we have $s>r$ due to our convention $s,r\leq n/(m+n)$. Thus the expression (\ref{expressiontoshowdminus1}) is strictly larger than $\Delta$, and because its left-hand-side is an increasing function of the argument $s\leq n/(m+n)$ (similar to the computation (\ref{paradigmaticallydifferentiateG2})), there exists due to continuity some $s'\in(r,s)$ with
\begin{align}\label{differenceistauinproofofthm1}
S(\tau_{s',m',n'})-S(\tau_{r,m',n'})~=~\Delta~.
\end{align}
Since $\binrel(s\|r)$ is strictly increasing in its first argument for $s\geq r$ (using $r>0$, which holds due to $\Delta<\log d$), we have $D_2(s'\|r)<D_2(s\|r)$, which contradicts the optimality of the pair $(\sigma,\rho)$. In the case $\Delta<0$, it is $s<r$ (again using the convention $s,r\leq n/(m+n)$) and (\ref{expressiontoshowdminus1}) is thus strictly smaller than $\Delta$. Therefore, we can find $r'\in(s,r)$ with $S(\tau_{s,m',n'})-S(\tau_{r',m',n'})=\Delta$, and it is now $D_2(s\|r')<D_2(s\|r)$ (irrespective of the value of $s$). We have thus shown that, if we choose the parametrization of the optimal pair in (\ref{diagstateswith2entries}) such that $s\leq n/(m+n)$, then there do \emph{not} exist $m',n'\geq1$ with $m'+n'\leq d$ and $n'/m'>n/m$. This implies $n=d-1$, $m=1$ for the optimal pair $(\sigma,\rho)$.

Using now $n=d-1$, $m=1$ in (\ref{diagstateswith2entries}) and recalling (\ref{entropyforstatesofspecialform})--(\ref{relativeentropyforstatesofspecialform}), the optimal states (for $\Delta\neq0$ and satisfying Eq.\ (\ref{lagrangeequationdifferentiated})) will thus be of the form (\ref{formofsigmaandrhoinDvsDeltaLemma}), where $(s,r)$ attains the minimum in (\ref{definefunctionM}); for $\Delta=0$, the optimal states can be chosen to be of that form.

\bigskip

So far, we have examined the (potentially) optimal states $(\widehat{\sigma},\widehat{\rho})$ satisfying Eq.\ (\ref{lagrangeequationdifferentiated}). We now show that the solutions of Eq.\ (\ref{linearlydependentgradients}) do not yield any new optimizing assignments. Condition (\ref{linearlydependentgradients}) holds iff there exists $\widehat{\lambda}\in(-\infty,+\infty)$ such that
\begin{align}
1+\log\widehat{p}_j~=~\frac{d}{dp_j}\left(\sum_ip_i\log p_i\right)\Bigg|_{\{\widehat{p}_i\}}~=~\widehat{\lambda}\,\frac{d}{dp_j}\left(\sum_ip_i\right)\Bigg|_{\{\widehat{p}_i\}}~=~\widehat{\lambda}\qquad\forall j~\text{with}~\,\widehat{p}_j\neq0~.
\end{align}
This holds iff $\log\widehat{p}_j=\log\widehat{p}_k$ whenever $\widehat{p}_j,\widehat{p}_k\neq0$, i.e.\ it holds exactly iff $\widehat{\rho}$ is completely mixed on its support. When $\supp[\widehat{\sigma}]\not\subseteq\supp[\widehat{\rho}]$, then $D(\widehat{\sigma}\|\widehat{\rho})=\infty$, and this is not a minimizing assignment except when $S(\widehat{\sigma})-S(\widehat{\rho})=\log d$, which is however already contained in the solutions (\ref{diagstateswith2entries}) found above with $n=d-1$, $m=1$. On the other hand, when $\supp[\widehat{\sigma}]\subseteq\supp[\widehat{\rho}]$ and $\widehat{\rho}$ is completely mixed on its support, one can compute
\begin{align}\label{entropydiffminusrelentropy}
\Delta~=~S(\widehat{\sigma})-S(\widehat{\rho})~=~-D(\widehat{\sigma}\|\widehat{\rho})~=~S(\widehat{\sigma})-\log{\rm rank}(\widehat{\rho})~\leq~0~.
\end{align}
Thus, we always have $D(\widehat{\sigma}\|\widehat{\rho})=- \Delta\in[0,\log d]$ here. Among these solutions, the cases $\Delta=0$ and $\Delta=-\log d$ have been discussed above and are contained in (\ref{diagstateswith2entries}) with $n=d-1$, $m=1$. We will finally show that all other solutions of (\ref{entropydiffminusrelentropy}) (and thus of (\ref{linearlydependentgradients})) are not minimizers of the optimization problem (\ref{stateinfimum}) by showing that for any $\Delta\in(-\log d,0)$ one can find states $\sigma$, $\rho$ with $S(\sigma)-S(\rho)=\Delta$ and $D(\sigma\|\rho)<-\Delta$. For this, let $\sigma:=\ket{\psi}\bra{\psi}$ be any fixed pure state and let $\rho_\mu:=\mu\ii_d/d+(1-\mu)\sigma$ for $\mu\in[0,1]$ be convex mixtures of the maximally mixed state $\ii_d/d$ with $\sigma$. Similar to Remark \ref{Mincreasingremark}, let $\mu'\in(0,1)$ be such that $S(\sigma)-S(\rho_{\mu'})=\Delta$, and notice again that $\mu'<-\Delta/\log d$ due to strict concavity of the entropy:\ $\Delta=S(\sigma)-S(\rho_{\mu'})<S(\sigma)-(\mu'S(\ii_d/d)+(1-\mu')S(\sigma))=-\mu'\log d$. Defining $\rho:=\rho_{\mu'}$, convexity of the relative entropy then indeed gives:
\begin{align}
D(\sigma\|\rho)~\leq~\mu'D(\sigma\|\ii_d/d)+(1-\mu')D(\sigma\|\sigma)~<~(-\Delta/\log d)\,D(\ket{\psi}\bra{\psi}\,\|\,\ii_d/d)~=~-\Delta~.
\end{align}
One may notice that all these better pairs $(\sigma=\ket{\psi}\bra{\psi},\rho_\mu)$ here are contained in the solutions (\ref{diagstateswith2entries}) found above with $n=d-1$, $m=1$.

\bigskip

The preceding proof shows also that, for $\Delta\neq0$, the optimal states are necessarily of the form (\ref{formofsigmaandrhoinDvsDeltaLemma}), up to simultaneous unitary transformations of $\sigma$ and $\rho$; the proof in Section \ref{proofofpropertiessubsect} shows furthermore that, for each $\Delta\neq0$, the optimal $s$ and $r$ are unique. For $\Delta=0$, the optimal pairs are obviously exactly the ones with $\sigma=\rho$.
\end{proof}

\subsection{Proof of Theorem \ref{propertiestheorem}}\label{proofofpropertiessubsect}
\begin{proof}[Proof of Theorem \ref{propertiestheorem}]
$M(\Delta,d)\geq0$ is clear, and the stated values are argued below Eq.\ (\ref{definebinaryrelent}). For the convenient upper bounds on $N(d)$, see Lemma \ref{lemmareducingonevariableoptimization}.

For $N=N(d)$, the first inequality in (\ref{lowerboundeqninmaintheorem}) is just Lemma \ref{lemmathatreducesoptimizationovertwovariables}, and for $N\geq N(d)$ it follows from the monotonicity of the lower bound:
\begin{align}
\frac{d}{dN}\left(Ne^\frac{\Delta}{N}-N-\Delta\right)~=~-e^\frac{\Delta}{N}\left[e^{-\frac{\Delta}{N}}-\left(1-\frac{\Delta}{N}\right)\right]~\leq~0~,
\end{align}
since the square brackets is non-negative due to convexity of the exponential function. For any $N$ and $\Delta$, the second inequality in (\ref{lowerboundeqninmaintheorem}) is easily verified by subtracting both sides from each other and observing that the difference and its first three derivatives w.r.t.\ $\Delta$ vanish at $\Delta=0$, whereas the fourth derivative is positive everywhere. If one defines, as usual, the minimum over an empty set in (\ref{definefunctionM}) to be $\infty$, then the lower bounds (\ref{lowerboundeqninmaintheorem}) hold even for $\Delta$ outside the range $[-\log d,\log d]$.

To prove (\ref{quadraticlowerboundforallDelta}) for $d\geq3$, we use the rightmost bound in (\ref{lowerboundeqninmaintheorem}) with $N=\log^2d$ and show $\Delta^2/(2\log^2d)+\Delta^3/(6\log^4d)\geq\Delta^2/(3\log^2d)$ for $\Delta\in[-\log d,\log d]$; this inequality is easily seen to hold whenever $\log d\geq1$. For $d=2$ and $\Delta\in[-\log^22,\log2]$ the last inequality holds as well; for $d=2$ and $\Delta\in[-\log2,-\log^22]$ we use the left inequality in (\ref{lowerboundeqninmaintheorem}) with $N=0.45>N(2)$ and verify numerically (cf.\ also upper left panel in Fig.\ \ref{Mfigure}) that $Ne^{\Delta/N}-N-\Delta\geq\Delta^2/(3\log^22)$ holds in this range of $\Delta$, with the gap in the inequality being at last $0.005$ which is well above $0$ numerically.

\bigskip

We now sketch a proof of strict convexity (and continuous differentiability) of $M(\Delta,d)$, which is somewhat involved; see also the proof of Theorem 1 in \cite{refinementspinskier} for a related approach at optimal refinements of Pinsker's inequality. For our proof, we employ the definition (\ref{definefunctionM}), will somtimes abbreviate $D:=\log(d-1)\geq0$, and denote by $r_d$ the (unique) $r\in(0,1/2)$ attaining the maximum in (\ref{definerealN}), i.e.\ satisfying $(1-2r_d)\log\left(\frac{1-r_d}{r_d}(d-1)\right)=2$. We also define $\gamma_d\in(0,(d-1)/d)$ to be the unique solution of $(1-\gamma_d)\log\left(\frac{1-\gamma_d}{\gamma_d}(d-1)\right)=1$; one can check that $\gamma_d>r_d$.

If, for some $\Delta=x\in(-\log d,\log d)$, a pair $(s,r)\in(0,(d-1)/d)^2$ attains the minimum in (\ref{definefunctionM}), then by the method of Lagrange multipliers the following two equations hold:
\begin{align}
\Delta(s,r)~&:=\binH(s)-\binH(r)+(s-r)D~=~x~,\label{DeltaSRx}\\
F(s,r)~&:=\left(\log\frac{1-r}{r}-\log\frac{1-s}{s}\right)\left(D+\log\frac{1-r}{r}\right)-\left(\frac{s}{r}-\frac{1-s}{1-r}\right)\left(D+\log\frac{1-s}{s}\right)~=~0~,\label{FSR0}
\end{align}
where the latter equality expresses the requirement that the gradients of the target function and the constraint function be parallel (i.e., that the $2\times2$-matrix formed by these gradients have vanishing determinant). In a small enough neighborhood of any such pair $(s,r)\in(0,(d-1)/d)^2$ with $s\neq r$, the equations (\ref{DeltaSRx})--(\ref{FSR0}) are sufficiently well-behaved to have a unique solution $(s(x'),r(x'))$ for any $x'\in(x-\varepsilon,x+\varepsilon)$, as the solution of the differential equations obtained from (\ref{DeltaSRx})--(\ref{FSR0}). For any $s=r$, (\ref{DeltaSRx})--(\ref{FSR0}) are satisfied with $x=0$ (corresponding to the trivial optimality cases $\sigma=\rho$), but near any such point there are no other pairs with $F(s,r)=0$ and $s\neq r$ (as one sees from a quadratic expansion of $F(s,r)$) with the exception of $s=r=r_d$:\ around $x=0$ and $s=r=r_d$, the equations (\ref{DeltaSRx})--(\ref{FSR0}) have a solution with $\dot{s}(x=0)=(1-2r_d)/3$, $\dot{r}(x=0)=-(1-2r_d)/6$ (overdots denote derivatives w.r.t.\ $x$), which can be seen by computing the third directional derivatives of $F(s,r)$ at this point.

Examining the equation $F(s,r)=0$ for $(s,r)\in(0,(d-1)/d)^2$ (by way of discussing $F(s,r)$ and its derivative $F_s(s,r)$ along each fixed $r$) and furthermore considering optimal pairs $(s,r)$ for any $\Delta=x$ in (\ref{definefunctionM}) on the boundary of $[0,(d-1)/d]^2$, one finds the following:\ for $r=0$, optimal pairs are obtained for $s=0$ and for $s=(d-1)/d$ (where $x=\log d$); for $0<r<r_d$, optimal pairs are obtained for $s=r$ and for one other value $s\in(r_d,(d-1)/d)$ (where $0<x<\log d$); for $r=r_d$, the only optimal pair is obtained for $s=r_d$ (where $x=0$); for $r_d<r<\gamma_d$, optimal pairs are obtained for one value $s\in(0,r_d)$ (where $x\in(\Delta_r,0)$, where we define $\Delta_r:=\Delta(s=0,r=\gamma_d)=1-D+\log\gamma_d\in(-\log d,1-\log d)$) and for $s=r$; for $\gamma_d\leq r\leq(d-1)/d$, optimal pairs are obtained for $s=0$ (where $x\in[-\log d,\Delta_r]$) and for $s=r$.

Combining this with the above differentiability result and defining $s(0):=r(0):=r_d$ for $x=0$ while disregarding the other optimal pairs with $s=r$, we get the following:\ for any $x\in[-\log d,\log d]\setminus\{0\}$ there exists exactly one optimal pair $(s(x),r(x))$ (i.e.\ with $\Delta(s(x),r(x))=x$), the curve $(s(x),r(x))$ is continuous in $x\in[-\log d,\log d]$, and differentiable in $x\in(\Delta_r,\log d)$. Thus already, $M(x,d)=D(s(x)\|r(x))$ is continuous in $x\in[-\log d,\log d]$ (with the usual convention $\lim_{x\nearrow\log d}M(x,d)=\infty=M(\log d,d)$).

We can now finally prove strict convexity of $M(x,d)$.\ First, for $x\in[-\log d,\Delta_r]$, it is $s(x)=0$. One can thus explicitly write $\Delta=-\binH(r)-Dr$ as a function of $M=M(x,d)=\binrel(s=0\|r)=-\log(1-r)$ in this range of $\Delta=x$; the function $\Delta=\Delta(M)$ is easily seen to be continuously differentiable, strictly decreasing and strictly convex in this range. Its inverse $M=M(\Delta,d)$ is thus strictly convex as well and continuously differentiable in $\Delta\in(-\log d,\Delta_r]$, and one can compute $dM/d\Delta\left|_{\Delta=\Delta_r}\right.=-1$ (and $dM/d\Delta\left|_{\Delta\searrow-\log d}\right.=-\infty$).

Second, for $x\in(\Delta_r,\log d)$, the optimal pairs $(s(x),r(x))\in(0,(d-1)/d)^2$ satisfy (\ref{DeltaSRx})--(\ref{FSR0}). We can thus compute
\begin{align}
\frac{d}{dx}M(x,d)~&=~\frac{d}{dx}\binrel(s(x)\|r(x))\\
&=~\left(\log\frac{1-r(x)}{r(x)}-\log\frac{1-s(x)}{s(x)}\right)\dot{s}(x)-\left(\frac{s(x)}{r(x)}-\frac{1-s(x)}{1-r(x)}\right)\dot{r}(x)\\
&=~\left(\log\frac{1-r(x)}{r(x)}-\log\frac{1-s(x)}{s(x)}\right)\left(D+\log\frac{1-s(x)}{s(x)}\right)^{-1}~,\label{eqngetridoffirstderivatives}
\end{align}
where in the last step we used (\ref{FSR0}) and the derivative of (\ref{DeltaSRx}) w.r.t.\ $x$. Notice for later that $dM(x,d)/dx\left|_{x\searrow\Delta_r}\right.=-1$ since $s(x)\searrow0$ and $r(x)\to\gamma_r$ for $x\searrow\Delta_r$. Thus,
\begin{align}
\left(D+\log\frac{1-s(x)}{s(x)}\right)^2\,\frac{d^2}{dx^2}M(x,d)~=~&\left(D+\log\frac{1-r(x)}{r(x)}\right)\frac{\dot{s}(x)}{s(x)(1-s(x))}\nonumber\\
&-\left(D+\log\frac{1-s(x)}{s(x)}\right)\frac{\dot{r}(x)}{r(x)(1-r(x))}~.
\end{align}
Strict convexity, $d^2M(x,d)/dx^2>0$, would thus follow from $\dot{s}(x)\geq0$ and $\dot{r}(x)\leq0$; to see the last implication, note that not both of $\dot{s}(x)$ and $\dot{r}(x)$ can vanish simultaneously because of $d\Delta(s(x),r(x))/dx=1>0$. The last insight also shows that $\dot{s}(x)\leq0$ and $\dot{r}(x)\geq0$ cannot both be true simultaneously unless $\dot{s}(x)=\dot{r}(x)=0$. It thus suffices now to show that $\dot{s}(x)$ and $\dot{r}(x)$ cannot both be simultaneously positive nor both be simultaneously negative. For $x=0$, this was remarked above. For $x\in(\Delta_r,\log d)\setminus\{0\}$, we show it in the following way.

Differentiating (\ref{FSR0}), one has
\begin{align}\label{withtwopositivecoefficients}
0~=~\frac{d}{dx}F(s(x),r(x))~=~F_s(s(x),r(x))\,\dot{s}(x)\,+\,F_r(s(x),r(x))\,\dot{r}(x)~.
\end{align}
The considerations of the equation $F(s,r)=0$ above show that $F_s(s(x),r(x))>0$ for $s(x)\neq r(x)$. Finally, the fact that $s(x)>r(x)$ implies $r(x)<r_d$ and the fact that $s(x)<r(x)$ implies $r(x)>r_d$ (see above) can be used, together with (\ref{FSR0}), to show $F_r(s(x),r(x))>0$ for $s(x)\neq r(x)$. (\ref{withtwopositivecoefficients}) then implies that not both of $\dot{s}(x)$ and $\dot{r}(x)$ can have the same sign.

$M(x,d)$ is thus strictly convex in $x\in(\Delta_r,\log d)$, as well as in $x\in[-\log d,\Delta_r]$. Since $M(x,d)$ is continuous with matching left-sided and right-sided derivatives at $x=\Delta_r$ (see above), it is strictly convex in the whole range $x\in[-\log d,\log d]$. Continuity of $(s(x),r(x))$ and Eq.\ (\ref{eqngetridoffirstderivatives}), together with the above considerations of the range $x\in[-\log d,\Delta_r]$, finally prove continuous differentiability of $M(x,d)$ in $x\in(-\log d,\log d)$.\end{proof}

\subsection{Auxiliary Lemmas}\label{technicallemmassubsection}
\begin{lemma}[Simple lower bound on $M(\Delta,d)$]\label{lemmathatreducesoptimizationovertwovariables}
For $2\leq d<\infty$ and $\Delta\in[-\log d,\log d]$, the quantity $M(\Delta,d)$ from Eq.\ (\ref{definefunctionM}) is bounded from below as follows:
\begin{align}
M(\Delta,d)~&\geq~N(d)\left(e^\frac{\Delta}{N(d)}-1-\frac{\Delta}{N(d)}\right)~,
\end{align}
where $N(d)$ is defined in Eq.\ (\ref{definerealN}).
\end{lemma}

\begin{proof}Define the function $\Delta(s,r):=\binH(s)-\binH(r)+(s-r)\log(d-1)$. To show Lemma \ref{lemmathatreducesoptimizationovertwovariables}, we will prove
\begin{align}\label{derivativeinprooftoprove}
G(s,r)~:=~\binrel(s\|r)\,-\,N(d)\left(e^\frac{\Delta(s,r)}{N(d)}-1-\frac{\Delta(s,r)}{N(d)}\right)~\geq~0
\end{align}
for all $s,r\in[0,(d-1)/d]$. The statement is easily verified for $r=0$, since $\binrel(s\|0)=+\infty$ unless $s=0$. We thus fix $r\in(0,(d-1)/d]$ from now on, so that $G(s,r)$ is a function of $s\in[0,(d-1)/d]$.

At $s=r$, the function $G(s=r,r)=0$ vanishes, as does its first derivative
\begin{align}
\left.\frac{d}{ds}G(s,r)\right|_{s=r}~=~\left.\log\frac{1-r}{r}-\log\frac{1-s}{s}-\left(e^\frac{\Delta(s,r)}{N(d)}-1\right)\log\left(\frac{1-s}{s}(d-1)\right)\right|_{s=r}~=~0~.
\end{align}
Furthermore, $G(s,r)$ is convex in $s\in[0,(d-1)/d]$ since, for $s\in(0,(d-1)/d]$,
\begin{align}
\frac{d^2}{ds^2}G(s,r)~=~e^\frac{\Delta(s,r)}{N(d)}\frac{1}{N(d)\,s(1-s)}\left[N(d)\,-\,r(1-r)\left(\log\left(\frac{1-r}{r}(d-1)\right)\right)^2\right]~\geq~0
\end{align}
as the term in square brackets is non-negative due to the definition of $N(d)$ in Eq.\ (\ref{definerealN}).

All of this together shows that, for each fixed $r\in[0,(d-1)/d]$, $G(s,r)$ attains its minimum $0$ at $s=r$, which finally proves (\ref{derivativeinprooftoprove}).
\end{proof}

\begin{lemma}[Simple bounds on $N(d)$]\label{lemmareducingonevariableoptimization}
For $d\geq2$, the optimization $N(d)$ from Eq.\ (\ref{definerealN}) satisfies the following bounds:
\begin{align}
N_d-1~=~\frac{1}{4}\log^2(d-1)~<~N(d)~&<~N_d~=~\frac{1}{4}\log^2(d-1)+1~,\label{strongerboundintechnicallemma}\\
N(d)~&<~\log^2d~,\label{log2dinequalityforNd}
\end{align}
where $N_d$ in the first inequality was defined in Eq.\ (\ref{defineapproxN}).
\end{lemma}
\begin{proof}To prove the upper bound in (\ref{strongerboundintechnicallemma}), we show that for all $r\in[0,1]$,
\begin{align}\label{needlhsinproofofsimpleubforNd}
0~<~\frac{1}{4}\log^2(d-1)+1\,-\,r(1-r)\left(\log\left(\frac{1-r}{r}(d-1)\right)\right)^2~.
\end{align}
For $r=0,1$ this is clear due to the convention $0\cdot\infty=0$ (or by continuity), and for $r=1/2$ it is easily verified. Let thus $r\in(0,1)\setminus\{1/2\}$. The right-hand-side of (\ref{needlhsinproofofsimpleubforNd}) equals
\begin{align}
&=~\left(\frac{1}{2}-r\right)^2\log^2(d-1)-2r(1-r)\left(\log\frac{1-r}{r}\right)\log(d-1)+1-r(1-r)\left(\log\frac{1-r}{r}\right)^2\nonumber\\
&=~\left(\left(\frac{1}{2}-r\right)\log(d-1)-\frac{r(1-r)}{\frac{1}{2}-r}\log\frac{1-r}{r}\right)^2+1-\left[r(1-r)+\frac{r^2(1-r)^2}{\left(\frac{1}{2}-r\right)^2}\right]\left(\log\frac{1-r}{r}\right)^2\nonumber\\
&\geq~\frac{1}{(1-2r)^2}\left[(1-2r)^2-r(1-r)\left(\log\frac{1-r}{r}\right)^2\right]~=:~\frac{\phi(r)}{(1-2r)^2}~,\label{showsquarebracketspositive}
\end{align}
where the inequality arises by omitting the non-negative first term $\left(\ldots\right)^2$ from the step before.

Now, the last expression does not depend on the dimension $d$ anymore, and one can show that it is positive for all $r\in(0,1)\setminus\{1/2\}$. This is numerically easily verified, or analytically in the following way:\ the term in square brackets in (\ref{showsquarebracketspositive}) vanishes at $r=1/2$, as do its first three derivatives w.r.t.\ $r$, whereas its fourth derivative
\begin{align}
\frac{d^4}{dr^4}\,\phi(r)~=~\frac{8}{r^2(1-r)^2}+\frac{2(1-2r)^2}{r^3(1-r)^3}+(1-2r)\left(\log\frac{1-r}{r}\right)\frac{16r(1-r)+4(1-2r)^2}{r^3(1-r)^3}\nonumber
\end{align}
is strictly positive for all $r\in(0,1)$, since $(1-2r)\log\frac{1-r}{r}\geq0$ for $r\in(0,1)$.

The lower bound in (\ref{strongerboundintechnicallemma}) follows by letting $r\to1/2$ in the definition (\ref{definerealN}) of $N(d)$.

In the range $d\geq4>e^{\sqrt{4/3}}\approx3.2$, the bound (\ref{log2dinequalityforNd}) follows from (\ref{strongerboundintechnicallemma}) due to $\frac{1}{4}\log^2(d-1)+1\leq\frac{1}{4}\log^2d+\frac{3}{4}\cdot\frac{4}{3}\leq\frac{1}{4}\log^2d+\frac{3}{4}\log^2d=\log^2d$. For $d=2,3$ the claim can be verified numerically (cf.\ also the lower right panel of Fig.\ \ref{Mfigure}).
\end{proof}

\subsection{Proof of Theorem \ref{maxvarianceinfothm}}\label{proofofmaxvariancesubsect}
\begin{proof}[Proof of Theorem \ref{maxvarianceinfothm}]
For fixed $d\geq2$, we maximize the expression on the LHS of (\ref{firstineqformaxvarianceinfo}) or (\ref{expressvariationofinformationwithpi}) over all probability distributions $\{p_i\}$ (i.e., spectra of $\rho$), which leads to the Lagrange function
\begin{align}\label{lagrangefunctionforvariance}
L(\{p_i\},\nu)~:=~\sum_ip_i(\log p_i)^2-\left(\sum_ip_i\log p_i\right)^2\,+\,\nu\sum_ip_i~,
\end{align}
with the Lagrange multiplier $\nu$ corresponding to the normalization $\tr{\rho}=1$. Assume now that $\{\widehat{p}_i\}$ (corresponding to the state $\widehat{\rho}$) attains the maximum of (\ref{expressvariationofinformationwithpi}) over all probability distributions $\{p_i\}$ (due to continuity and compactness, this maximum is attained). We now view (\ref{lagrangefunctionforvariance}) as a function of those variables $p_i$ for which $\widehat{p}_i>0$, fixing the other elements $p_i$ to be zero. Then, due to the extremality of $\{\widehat{p}_i\}$ and having components in the interior of the domain of $L$, the method of Lagrange multipliers guarantees the existence of $\widehat{\nu}\in(-\infty,+\infty)$ such that
\begin{align}\label{onlytwononzerocomponentsinvariance}
\begin{split}
0~&
=~\left.\frac{dL}{dp_j}\right|_{\{\widehat{p}_i\},\widehat{\nu}}~=~(\log\widehat{p}_j)^2+2\log\widehat{p}_j-2\left(\sum_i\widehat{p}_i\log\widehat{p}_i\right)(1+\log\widehat{p}_j)+\widehat{\nu}\\
&=~\left(S(\{\widehat{p}_i\})+1+\log\widehat{p}_j\right)^2\,-\,\left(S(\{\widehat{p}_i\})\right)^2\,+\,\widehat{\nu}-1~~~\qquad\forall j~\,\text{with}~\,\widehat{p}_j>0~,
\end{split}
\end{align}
where the quantity $S(\{\widehat{p}_i\})=S(\widehat{\rho})$ denotes the entropy of the distribution $\{\widehat{p}_i\}$ and in particular does not depend on the index $j$. Thus, the equality (\ref{onlytwononzerocomponentsinvariance}) implies that
\begin{align}
\log\widehat{p}_j~=~\pm\sqrt{\left(S(\widehat{\rho})\right)^2-\widehat{\nu}+1}\,-\,S(\widehat{\rho})-1~~~\qquad\forall j~\,\text{with}~\,\widehat{p}_j>0~,
\end{align}
so that strict monotonicity of the logarithm yields that there can be at most two distinct non-zero elements in $\{\widehat{p}_i\}$.

Thus, leaving off hats again, an optimal $\rho=\widehat{\rho}$ has the form
\begin{align}
{\rho}~=~{\rm diag}\left(\frac{1-{r}}{m},\ldots,\frac{1-{r}}{m},\frac{{r}}{n},\ldots,\frac{{r}}{n},0,\ldots,0\right)
\end{align}
with $m,n\geq1$, $m+n\leq d$, $r\in[0,1]$. W.l.o.g.~we can assume $r\leq1/2$ by permuting the entries of $\rho$. For such states one has, after a small calculation,
\begin{align}\label{maximizerinproof}
\var_\rho(\log\rho)~=~r(1-r)\left(\log\frac{1-r}{r}+
\log\frac{n}{m}\right)^2~.
\end{align}
Maximizing this, for any fixed $r\in[0,1/2]$, over $m$ and $n$ yields $n=d-1$ and $m=1$. Maximizing (\ref{maximizerinproof}) finally over $r$ gives a unique $r=r_d\in(0,1/2)$, namely the unique value of $r\in[0,1/2]$ satisfying $(1-2r)\log\left(\frac{1-r}{r}(d-1)\right)=2$, and the maximum of (\ref{maximizerinproof}) is $N(d)$ from Eq.\ (\ref{definerealN}).

The last inequality in (\ref{firstineqformaxvarianceinfo}) is shown by Lemma \ref{lemmareducingonevariableoptimization}, completing the proof of Theorem \ref{maxvarianceinfothm}.
\end{proof}

\bigskip\medskip{\bf Acknowledgments.} We would like to thank Daniel Reitzner and Marco Tomamichel for helpful discussions. DR was supported by the Marie Curie Intra European Fellowship QUINTYL. MMW acknowledges support from the Alfried Krupp von Bohlen und Halbach-Stiftung.

\bibliographystyle{alpha}

\end{document}